\documentclass[screen,acmsmall,nonacm]{acmart}

\usepackage[capitalize]{cleveref}
\usepackage{code}
\usepackage{macros}
\usepackage{ebproof}

\usepackage{natbib}

\AtEndPreamble{
  \theoremstyle{acmplain}
  \newtheorem{axiom}[theorem]{Axiom}

  \theoremstyle{acmdefinition}
  \newtheorem{construction}[theorem]{Construction}
  \newtheorem{computation}[theorem]{Computation}
  \newtheorem{notation}[theorem]{Notation}
  \newtheorem{remark}[theorem]{Remark}
}

\setcopyright{acmcopyright}
\copyrightyear{2021}
\acmYear{2021}

\acmJournal{JACM}


\title{Logical Relations as Types}
\subtitle{Proof-Relevant Parametricity for Program Modules}

\author{Jonathan Sterling}
\orcid{0000-0002-0585-5564}
\email{jmsterli@cs.cmu.edu}

\author{Robert Harper}
\orcid{0000-0002-9400-2941}
\email{rwh@cs.cmu.edu}

\affiliation{
  \institution{Carnegie Mellon University}
  \city{Pittsburgh}
  \state{Pennsylvania}
  \country{USA}
}

\begin{abstract}
The theory of program modules is of interest to language designers not only for
its practical importance to programming, but also because it lies at the nexus
of three fundamental concerns in language design: \emph{the phase distinction},
\emph{computational effects}, and \emph{type abstraction}.  We contribute a
fresh ``synthetic'' take on program modules that treats modules as the
fundamental constructs, in which the usual suspects of prior module calculi
(kinds, constructors, dynamic programs) are rendered as derived notions in
terms of a modal type-theoretic account of the phase distinction.  We simplify
the account of type abstraction (embodied in the generativity of module
functors) through a \emph{lax modality} that encapsulates computational
effects, placing \emph{projectibility} of module expressions on a
type-theoretic basis. 

Our main result is a (significant) proof-relevant and phase-sensitive
generalization of the Reynolds abstraction theorem for a calculus of program
modules, based on a new kind of logical relation called a \emph{parametricity
structure}.  Parametricity structures generalize the proof-irrelevant relations
of classical parametricity to proof-\emph{relevant} families, where there may
be non-trivial evidence witnessing the relatedness of two programs ---
simplifying the metatheory of strong sums over the collection of types, for
although there can be no ``relation classifying relations'', one easily
accommodates a ``family classifying small families''.

Using the insight that logical relations/parametricity is \emph{itself} a
form of phase distinction between the syntactic and the semantic, we
contribute a new synthetic approach to phase separated parametricity based on
the slogan \emph{logical relations as types}, by iterating our modal account
of the phase distinction. We axiomatize a dependent type theory of
parametricity structures using two pairs of complementary modalities
(syntactic, semantic) and (static, dynamic), substantiated using the topos
theoretic \emph{Artin gluing} construction. Then, to construct a simulation
between two implementations of an abstract type, one simply programs a third
implementation whose type component carries the representation invariant.

 \end{abstract}

\begin{document}
\maketitle

\section{Introduction}\label{sec:intro}

Program modules are the application of dependent type theory with universes to
the large-scale structuring of programs. As \citet{macqueen:1986} observed, the
hierarchical structuring of programs is an instance of dependent sum; consider
the example of a type together with a pretty printer:

\begin{code}
    (* SHOW $:=$ $\Sum{T:\Univ}{\prn{T\rightharpoonup\mathcd{string}}}$ *)
    signature SHOW =
    sig
      type t
      val show : t $\rightharpoonup$ string
    end
\end{code}

On the other hand, the parameterization of a program component in
another component is an instance of dependent product; for instance, consider a
module functor that implements a pretty printer for a product type:
\begin{code}
    (* ShowProd $:$ $\Prod{S_1,S_2:\mathcd{SHOW}}{\prn{\pi_1\prn{S_1}*\pi_1\prn{S_2}\rightharpoonup\mathcd{string}}}$ *)
    functor ShowProd (S1 : SHOW) (S2 : SHOW) :
    sig
      type t = S1.t * S2.t
      val show : t $\rightharpoonup$ string
    end = ...
\end{code}

Modules are more than just dependent products, sums, and universes, however: a
module language must account for abstraction and the phase
distinction, two critical notions that seem to complicate the simple story of
modules as dependent types. In \cref{sec:modal-abstraction}, we
introduce \ModTT{}, our take on a type theory for program modules, and explain
how to view abstraction and generativity in terms of a \emph{lax modality} or
strong monad; in \cref{sec:phase-distinction}, the phase distinction is seen to
arise naturally from an \emph{open modality} in the sense of topos theory.

\subsection{Abstraction and computational effects}\label{sec:modal-abstraction}

Reynolds famously argued that \emph{``Type structure is a syntactic discipline
for enforcing levels of abstraction''}~\citep{reynolds:1983}; abstraction is
the facility to manage the non-equivalence of types at the boundary
between spuriously compatible program fragments --- for instance, the boundary
between a fragment of a compiler that emits a De Bruijn index (address of a variable counted
from the right) and a fragment that accepts a De Bruijn level (the address counted from the left).

\subsubsection{Static abstraction via let binding}

The primary aspect of abstraction is, then, to prevent the ``false linkage'' of
programs permitted by coincidence of representation; the static distinction
between two different uses of the same type can be achieved by the standard
rule for (non-dependent) let-binding in type theory:
\[
  \inferrule[non-dependent let]{
    \Gamma\vdash A\ \mathit{type}\\
    \Gamma\vdash B\ \mathit{type}\\
    \Gamma\vdash N : A\\
    \Gamma, x : A \vdash M : B
  }{
    \Gamma \vdash \prn{\mathsf{let}\ x : A = N\ \mathsf{in}\ M} : B
  }
\]

Static ``let abstraction'' as above enables the programmer to treat the same
type differently in two locations, but share the same values at runtime. For
instance, consider the following expression that binds the \emph{integer
equality} structure twice, for two different purposes:
\[
  \begin{array}{l}
    \mathsf{let}\ \mathsf{DeBruijnLevel} : \mathsf{EQ} = \prn{\mathsf{int},\mathsf{int\_eq}}\ \mathsf{in}\\
    \mathsf{let}\ \mathsf{DeBruijnIndex} : \mathsf{EQ} = \prn{\mathsf{int},\mathsf{int\_eq}}\ \mathsf{in}\\
    M
  \end{array}
\]

In the scope of $M$ it is not the case that $\mathsf{DeBruijnLevel}$ and
$\mathsf{DeBruijnIndex}$ have the same type component. But at runtime, $M$ will
be instantiated with the same type and value components in both positions.  In
the Standard~ML implementation of modules, a more sophisticated form of
let binding is elaborated that actually exposes the static identity of the
bound term in the body; for this reason, Standard~ML programmers use
\emph{dynamic} abstraction (\cref{sec:dynamic-abstraction}) via the opaque
ascription \textcd{M :> S} to negotiate both static and dynamic abstraction situations.

\subsubsection{Dynamic abstraction via modal binding}\label{sec:dynamic-abstraction}

In the presence of computational effects and module functors, it is not always
enough to statically distinguish between two ``instances'' of the same type:
the body of a module functor may contain a local state that must be distinctly
initiated in every instantiation. Sometimes referred to as \emph{generativity},
the need for this dynamic form of abstraction can be illustrated by means of an
ephemeral structure to manage a given namespace in a compiler:

\begin{code}
    signature NAMESPACE =
    sig
      type symbol

      val defined : string $\rightharpoonup$ bool

      val into : string $\rightharpoonup$ symbol
      val out : symbol $\rightharpoonup$ string
      val eq : symbol * symbol $\rightharpoonup$ bool
    end
\end{code}

\begin{figure}
  \begin{code}
    functor Namespace (A : ARRAY) :> NAMESPACE =
    struct
      type symbol = int

      val table = A.new (* allocation size *)
      val defined str = (* see if [str] has already been allocated *)
      val into str = (* hash [str] and insert it into [table] if needed *)
      fun out sym =
        case A.sub (table, sym) of
        | NONE $\Rightarrow$ raise Impossible
        | SOME str $\Rightarrow$ str
    end
  \end{code}
  \caption{A functor that generates a new namespace in Standard~ML.}
\end{figure}

To manage two different namespaces, one requires two distinct copies
\textcd{NS1, NS2} of the \textcd{Namespace} structure. If it were not for the
\textcd{defined} operator, it would be safe to generate a single
\textcd{Namespace} structure and bind it to two different module variables: we
would have \textcd{NS1.symbol $\not=$ NS2.symbol} but at runtime, the same table would be
used. However, this behavior becomes observably incorrect in the presence of
\textcd{defined}, which exposes the internal state of the namespace.

The dynamic effect of initializing the namespace structure once per
instantiation has historically been treated in terms of a notion of
\emph{projectibility}~\citep{dreyer-crary-harper:2003,harper:2016}, restricting
when the components of a module expression can be projected; under the
generative semantics of module functors, a functor application is never
projectible. Projectibility, however, is not a type-theoretic concept because it
does not respect substitution!

We argue that it is substantially simpler to present the module calculus with
an explicit separation of effects via a lax modality / strong monad
$\SigCmp$; concurrent work of Crary supports the same
conclusion~\citep{crary:2020}. \ModTT{} distinguishes between commands
$M\mathrel{\div} \sigma$ and values $V : \sigma$, and mediates between them
using the standard rules of the lax modality~\citep{fairtlough-mendler:1997}:
\begin{mathpar}
  \inferrule{
    \IsSig{\Gamma}{\sigma}
  }{
    \IsSig{\Gamma}{\SigCmp{\sigma}}
  }
  \and
  \inferrule{
    \IsMod{\Gamma}{V}{\sigma}
  }{
    \IsCmp{\Gamma}{\CmpRet{V}}{\sigma}
  }
  \and
  \inferrule{
    \IsMod{\Gamma}{V}{\SigCmp{\sigma}}\\
    \IsCmp{\Gamma,X:\sigma}{M}{\sigma'}
  }{
    \IsCmp{\Gamma}{\prn{\CmpBind{X}{V}{M}}}{\sigma'}
  }
  \and
  \inferrule{
    \IsCmp{\Gamma}{M}{\sigma}
  }{
    \IsMod{\Gamma}{\ModCmp{M}}{\SigCmp{\sigma}}
  }
\end{mathpar}

In this style, one no longer needs the notion
of projectibility: a generative functor is nothing more than a module-level
function $\sigma\Rightarrow\SigCmp{\tau}$, and the result of applying such a function
must be bound in the monad before it can be used, so one naturally
obtains the generative semantics without resorting to an ad hoc notion of
``generative'' or ``applicative'' function space.

\begin{code}
    NS1 $\leftarrow$ Namespace (Array);
    NS2 $\leftarrow$ Namespace (Array); ...
\end{code}

\subsection{The phase distinction}\label{sec:phase-distinction}

The division of labor between the lightweight syntactic verification provided
by type abstraction and the more thoroughgoing but expensive verification
provided by program logics is substantiated by the phase distinction
between the static/compiletime and dynamic/runtime parts of a program respectively.
Respect for the phase distinction means that there is a well-defined notion of
static equivalence of program fragments that is independent of dynamic
equivalence; moreover, one must ensure that static equivalence is efficiently
decidable for it to be useful in practice.

\subsubsection{Explicit phase distinction}\label{sec:explicit-phase-distinction}

The phase distinction calculi of \citet{moggi:1989,harper-mitchell-moggi:1990}
capture the separation of static from dynamic in an explicit and intrinsic way: a core
calculus of modules is presented with an explicit distinction between (modules,
signatures) and (constructors, kinds) in which the latter play the role of the
static part of the former.
A signature is explicitly split into a (static) kind $k:\mathit{kind}$ and a
(dynamic) type $u : k \vdash t(u):\mathit{type}$ that depends on it, and
module value is a pair $(c,e)$ where $c : k$ and $e : t(c)$. Functions of modules
are defined by a ``twinned'' lambda abstraction $\lambda u/x.M$, and scoping
rules are used to ensure that static parts depend only on constructor variables
$u:k$ and not on term variables $x:t$.

An unfortunate consequence of the explicit presentation of phase separation is
that the rules for type-theoretic connectives (dependent product, dependent
sum) become wholely non-standard and it is not immediately clear in which sense
these actually \emph{are} dependent product or sum. For instance, one has rules
like the following for dependent product:

{
  \small
  \[
    \inferrule[pi formation*]{
      \Delta\vdash k\ \mathit{kind}\\
      \Delta,u:k \vdash k'(u)\ \mathit{kind}\\\\
      \Delta,u:k;\Gamma\vdash \sigma\prn{u}\ \mathit{type}\\
      \Delta,u:k;\Gamma,u':k'\prn{u};\Gamma \vdash \sigma'\prn{u,u'}\ \mathit{type}\\
    }{
      \Delta;\Gamma\vdash \Pi u/X : \brk{u:k. \sigma\prn{u}}. \brk{u':k'\prn{u}. \sigma'\prn{u,u'}} \equiv
      \brk{k : \prn{\Pi u : k. k'\prn{u}}; \Pi u : k. \sigma\prn{u} \to \sigma'\prn{u, v\prn{u}}}\
      \mathit{sig}
    }
  \]
}

\paragraph{The Grothendieck construction}

\citeauthor{moggi:1989} observed that the explicit phase distinction calculus
can be understood as arising from an indexed category in the following sense:
\begin{enumerate}

  \item One begins with a purely static language, \ie a category $\BCat$ whose
    objects are kinds and whose morphisms are constructors.

  \item Next one defines an indexed category
    $\Mor[\CCat]{\OpCat{\BCat}}{\CAT}$: for a kind $k$, the fiber category
    $\CCat\prn{k}$ is the collection of signatures with static part $k$, with
    morphisms given by functions of module expressions.

\end{enumerate}

Then, the syntactic category of the full calculus is obtained by the
\emph{Grothendieck construction} $\GCat = \int\Sub{\BCat}\CCat$, which takes an
indexed category to its total category. An object of $\GCat$ is a pair
$\prn{k, \sigma}$ with $k : \BCat$ and $\sigma : \CCat\prn{k}$; a morphism
$\Mor{\prn{k,\sigma}}{\prn{k',\sigma'}}$ is a morphism $\Mor[c]{k}{k'}:\BCat$
together with a morphism $\Mor{\sigma}{c^*\sigma'} : \CCat\prn{k}$, where, as
usual, $c^{\ast}$ is $\CCat{c}$.

The benefit of considering $\GCat$ is that the non-standard rules for type
theoretic connectives become a special case of the standard ones: from this
perspective, the strange \textsc{pi formation*} rule (with its nonstandard
contexts and scoping and variable twinning) above can be seen to be a certain
calculation in the Grothendieck construction of a certain dependent
product.

\subsubsection{Implicit phase distinction}

An alternative to the explicit phase separation of
\citet{harper-mitchell-moggi:1990} is to treat the module calculus as
ordinary type theory, extended by a judgment for \emph{static
equivalence}. Then, two modules are considered statically equivalent when they
have the same static part --- though the projection of static parts is defined
metatheoretically rather than intrinsically. This approach is represented by
\citet{dreyer-crary-harper:2003}.

\subsubsection{This paper: synthetic phase distinction}\label{sec:synthetic-phase-distinction}

Taking inspiration from both the explicit and implicit accounts of phase
separation, we note that the detour through indexed categories was strictly
unnecessary, and the object of real interest is the category $\GCat$ and the
corresponding fibration $\FibMor{\GCat}{\BCat}$ that projects the static
language from the full language. We obtain further leverage by additionally specifying
$\BCat$ as a slice $\Sl{\GCat}{\StOpn}$ for a special object $\StOpn : \GCat$.
In the phase-split setting, the object $\StOpn$ corresponds to a signature
$\prn{\_ : \top.\bot}$ whose static part is terminal and whose dynamic part is
initial; the intuition behind this definition is that the presence of $\bot$ at
the dynamic level ``zeroes out'' any dynamic data to its right, whereas $\top$
at the static level has no effect.

The view of $\BCat$ as a slice of $\GCat$ is inspired by Artin
gluing~\citep{sga:4}, a mathematical version of logical predicates in which
the syntactic category of a theory is reconstructed as a slice of a topos of
logical predicates: there is a very precise sense in which the notion of
``signature over a kind'' can be identified with ``logical predicate on a
kind''.
The connection between phase separation and gluing/logical predicates is, to
our knowledge, a novel contribution of this paper.

Put syntactically, the language corresponding to $\GCat$ possesses a new
context-former $\prn{\Gamma, \StOpn}$ called the ``static open'';\footnote{The
terminology of ``opens'' is inspired by topos theory, in which proof irrelevant
propositions correspond to partitions into open and closed subtopoi. Indeed,
such a partition is the geometrical prototype of the phase distinction, an insight that informs the central tool of this paper.} when
$\StOpn$ is in the context, everything except the static part of an object is ignored by the judgmental equality relation $A\equiv B$.
For instance, module commands and terms of program type are rendered purely
dynamic / statically inert by means of special rules of static
connectivity under the assumption of $\StOpn$:
\begin{mathpar}
  \inferrule[static open]{
    \IsCx{\Gamma}
  }{
    \IsCx{\Gamma,\StOpn}
  }
  \and
  \inferrule[static connectivity (1)]{
    \IsMod{\Gamma}{t}{\SigTp}\\
    \Gamma\vdash \StOpn
  }{
    \IsMod{\Gamma}{*}{t}
  }
  \and
  \inferrule[static connectivity (2)]{
    \IsMod{\Gamma}{t}{\SigTp}\\
    \IsMod{\Gamma}{e}{t}\\
    \Gamma\vdash \StOpn
  }{
    \EqMod{\Gamma}{e}{*}{t}
  }
  \and
  \inferrule[static connectivity (3)]{
    \IsSig{\Gamma}{\sigma}\\
    \Gamma\vdash \StOpn
  }{
    \IsCmp{\Gamma}{*}{\sigma}
  }
  \and
  \inferrule[static connectivity (4)]{
    \IsSig{\Gamma}{\sigma}\\
    \IsCmp{\Gamma}{M}{\sigma}\\
    \Gamma\vdash \StOpn
  }{
    \EqCmp{\Gamma}{M}{*}{\sigma}
  }
\end{mathpar}

\paragraph{Signatures, kinds, and static equivalence}

In our account, the phase distinction between signatures/modules and
kinds/constructors is expressed by a universal property: a signature
$\Gamma\vdash\sigma\ \mathit{sig}$ is called a kind iff the weakening of sets
of equivalence classes from $\brc{\brk{V}\mid \Gamma\vdash V : \sigma}$ to
$\brc{\brk{V}\mid \Gamma,\StOpn\vdash V:\sigma}$ is an isomorphism natural in
$\Gamma$. In other words, the exponentiation by $\StOpn$ defines an \emph{open
modality} $\MSt = \prn{\IHom{\StOpn}{-}}$ in the sense of topos theory.

Because the modality $\MSt$ is idempotent, we may define (internally!) the
static part of any signature $\sigma$ as $\MSt{\sigma}$; the
modal unit $\Mor[\upeta\Sub{\MSt}]{\sigma}{\MSt{\sigma}}$ abstractly implements
the projection of constructors from module values. Because the modality $\MSt$
is defined by exponentiation with a subterminal (\ie a proof-irrelevant
sort), it is easy to show internally that the usual equations of static
projection hold (naturally, up to isomorphism): for instance, we have $\MSt{\sigma\Rightarrow \tau} \cong
\MSt{\sigma}\Rightarrow\MSt{\tau}$, \etc.

The notion of static equivalence from \citet{dreyer-crary-harper:2003} is then
reconstructed as ordinary judgmental equality in the context of $\StOpn$;
the view of phase separation as a projection functor from \citet{moggi:1989} is
reconstructed by the weakening $\Mor{\GCat}{\Sl{\GCat}{\StOpn}}$.

\subsection{Sharing constraints, singletons, and the \emph{static extent} connective}

An important practical aspect of module languages is the ability to constrain the
identity of a substructure; for instance, the implementation of IP in the
FoxNet protocol stack~\citep{foxnet:1994} is given as a functor taking two structures as arguments
\emph{under the additional constraint} that the structures have compatible type components:

\begin{code}
    functor Ip
      (structure Lower : PROTOCOL
       structure B : FOX\_BASIS
         where type Receive\_Packet.T = Lower.incoming\_message
       ...)
\end{code}

\subsubsection{Sharing as pullback}

The above fragment of the input to the \textcd{Ip} functor can be viewed as a pullback of two
signatures along type projections, rather than a product of two signatures:
\[
  \DiagramSquare{
    nw/style = pullback,
    width = 5cm,
    height = 1.5cm,
    ne = \mathcd{PROTOCOL},
    sw = \mathcd{FOX\_BASIS},
    se = \mathcd{type},
    south = \mathcd{.Receive\_Packet.T},
    east = \mathcd{.incoming\_message},
    north = \mathcd{.Lower},
    west = \mathcd{.B}
  }
\]

The view of sharing in terms of pullback or equalizers, proposed by
\citet{mitchell-harper:1988}, is perfectly appropriate from a semantic
perspective; however, it unfortunately renders type
checking undecidable~\citep{castellan-clairambault-dybjer:2017}. Because types
in ML-style languages are meant to provide \emph{lightweight} verification, it
is essential that the type checking problem be tractable: therefore, something
weaker than general pullbacks is required. Semantically speaking, what one
needs is roughly pullback along display maps only, \ie equations that can
be oriented as definitions.

\subsubsection{Type sharing via singletons}

A strategy more well-adapted to implementation is to elaborate type sharing in
a way that involves a new singleton type signature $\mathcal{S}\prn{t}\
\mathit{sig}$ for each $t : \mathcd{type}$, as pioneered
by~\citet{harper-stone:2000}.  There is up to judgmental equality exactly
one module of signature $\mathcal{S}\prn{t}$, namely $t$ itself; in contrast to
general pullbacks, the singleton signature does not disrupt the decidability of
type equivalence~\citep{stone-harper:2006,abel-coquand-pagano:2009}.

The truly difficult part of singleton types, dealt with by
\citet{stone-harper:2006}, is their subtyping and re-typing principles: not
only should it be possible to pass from a more specific type to a less specific
type, it must also be possible to pass from a less specific type to a more
specific type when the identity of the value is known. Because of the
dependency involved in the latter transition, ordinary subtyping is not
enough to account for the full expressivity of singletons, hence the
extensional retyping principles of earlier work on singleton
calculi~\citep{dreyer-crary-harper:2003,crary:2019}.

As a basic principle, we do not treat subtyping or retyping directly in the
core type theory: we intend to give an \emph{algebraic} account of program
modules, so both subtyping and retyping become a matter of elaborating
coercions.  We propose to account for both the subtyping and retyping
principles via an elaboration algorithm guided by the $\eta$-laws of each
connective, including the $\eta$-laws of the singleton type connective.
Early evidence that our proposal is tractable can be found in the
implementation of the \texttt{\textcolor{RegalBlue}{cool}tt} proof assistant for cubical type theory, which treats a
generalization of singleton types via such an
algorithm~\citep{cooltt:2020}.\footnote{An example of the application of
\texttt{\textcolor{RegalBlue}{cool}tt}'s elaboration algorithm to the subtyping and retyping of
singletons can be found here:
\url{https://github.com/RedPRL/cooltt/blob/7be1bb32f8b0eaae75c5a11f1c1c5b0ff1086c94/test/selfification.cooltt}.}

\subsubsection{General sharing via the \emph{static extent}}

It is useful to express the compatibility of components of modules other than types:
families of types (e.g. the polymorphic type of lists)
are one example, but arguably one should be able to express a sharing
constraint on an entire substructure. Type theoretically, it is trivial to
generalize the type singletons in this direction, but we risk incurring static
dependencies on dynamic components of signatures, violating the spirit of the
phase distinction.

One of the design constraints for module systems, embodied in the phase
distinction, is that dependency should only involve static constructs;
the decidable fragment of the dynamic algebra of programs is
unfortunately too fine to act as more than an obstruction to the composition of
program components.
From our synthetic view of the phase distinction, it is most natural to rather
generalize the type singletons to a signature connective $\SigWhere{\sigma}{V}$
that classifies the ``static extent'' of a module $V:\sigma$ for an arbitrary
signature $\sigma$, summarized in the following rules of
inference:\footnote{For simplicity, we present these rules in a style that
violates uniqueness of types; the actual encoding in the logical framework is
achieved using explicit introduction and elimination forms.}
\begin{mathpar}
  \inferrule[formation]{
    \IsSig{\Gamma}{\sigma}\\\\
    \IsMod{\Gamma, \StOpn}{V}{\sigma}
  }{
    \IsSig{\Gamma}{\SigWhere{\sigma}{V}}
  }
  \and
  \inferrule[introduction]{
    \IsMod{\Gamma}{U}{\sigma}\\\\
    \EqMod{\Gamma,\StOpn}{U}{V}{\sigma}
  }{
    \IsMod{\Gamma}{U}{\SigWhere{\sigma}{V}}
  }
  \and
  \inferrule[elimination]{
    \IsMod{\Gamma}{U}{\SigWhere{\sigma}{V}}
  }{
    \IsMod{\Gamma}{U}{\sigma} \\
    \EqMod{\Gamma,\StOpn}{U}{V}{\sigma}
  }
\end{mathpar}

In \ModTT{}, the elements of the \emph{static extent} of a module $V:\sigma$ are all
the modules whose static part is judgmentally equal to $V$; therefore
$\SigWhere{\sigma}{V}$ is not a singleton in general, but it is a
singleton when $\sigma$ is purely static. Our approach is equivalent to (but
arguably more convenient than) the use of singleton kinds: the static extent
is admissible under the explicit phase distinction.

\paragraph{Extension types in cubical type theory}

Our static extent connective is inspired by the \emph{extension types} of
\citet{riehl-shulman:2017}, already available in a few implementations of
cubical type theory~\citep{redtt:2018,cooltt:2020}.
Whereas in cubical type theory one extends along a cofibrant subobject
$\Mor|>->|{\phi}{\mathbb{I}^n}$ of a cube, in a phase
separated module calculus one extends along the open domain
$\Mor|>->|{\StOpn}{\mathbf{1}}$. The static extent connective is also closely
related to the formal disk bundle of \citet{wellen:2017}, which
classifies the ``infinitesimal extent'' of a given point in synthetic
differential (higher) geometry.

\paragraph{Strong structure sharing \`a la SML~'90}
Another account of the sharing of structures is argued for in earlier versions
of Standard~ML~\citep{milner-tofte-harper:1990}, in which each structure is in
essence tagged with a static identity~\citep{macqueen-harper-reppy:2020}; this
``strong'' structure sharing was replaced in SML~'97 by the current ``weak''
structure sharing, which has force only on the static components of the
signature~\citep{milner-tofte-harper-macqueen:1997}. Our static extents capture exactly the semantics of weak
structure sharing; we note that the strong sharing of
SML~'90 can be simulated by adding a dummy abstract type to each signature
during elaboration.

\subsection{Proof-relevant parametricity: the objective metatheory of ML modules}

We outline an approach to the definition and metatheory of a calculus for
program modules, together with a modernized take on logical relations /
Tait computability that enables succinct proofs of representation independence
and parametricity results.

\subsubsection{Algebraic metatheory in an equational logical framework}\label{sec:algebraic-metatheory}

Many existing calculi for program modules are formulated using raw terms, and
animated via a mixture of judgmental equality (for the module layer) and
structural operational semantics (for the program layer). In contrast, we
formulate \ModTT{} entirely in an equational logical framework,\footnote{Though
we present it using standard notations for readability.} eschewing raw terms
entirely and \emph{only} considering terms up to typed judgmental equality.
Because we have adopted a modal separation of effects
(\cref{sec:modal-abstraction}), there is no obstacle to accounting for genuine
computational effects in the program layer, even in the purely equational
setting~\citep{staton:2013}.

The mechanization of Standard
ML~\citep{lee-crary-harper:2007,crary-harper:2009} in the Edinburgh Logical
Framework~\citep{harper-honsell-plotkin:1993} is an obvious precursor to our
design; whereas in the cited work, the LF's function space was used to encode
the binding structure of raw terms and derivations, we employ the internal language of
locally Cartesian closed categories as a logical framework to account for both
typing and judgmental equality of abstract terms, as explicated by
\citet{gratzer-sterling:2020}. The idea of dependently typed equational logical
frameworks goes back to \citet{cartmell:1978} (for theories without binding),
and was further developed by Martin-L\"of for theories with binding of
arbitrary order~\citep{nordstrom-peterson-smith:1990}.
Because we work only with typed terms up to judgmental equality, we may use
\emph{semantic} methods such as Artin gluing to succinctly prove syntactic
results as in several recent
works~\citep{altenkirch-kaposi:2016:nbe,coquand:2019,kaposi-huber-sattler:2019,coquand-huber-sattler:2019,
sterling-angiuli-gratzer:2019,sterling-angiuli-gratzer:2020,sterling-angiuli:2020}.

The effectiveness of algebraic methods relies on the existence of initial
algebras for theories defined in a logical framework. The existence of initial
algebras is not hard to prove and usually follows from
standard results in category theory.
That an initial algebra can be \emph{presented} by a quotient of raw syntax is
more laborious to prove for a given logical framework (see
\citet{streicher:1991} for a valiant effort); such a result is the combination
of soundness and completeness.

It comes as a pleasant surprise, then, that the syntactic presentation of the
core language is not in practice germane to the study of real type theories
and programming languages: the only raw syntax one need be concerned with is
that of the surface language, but the surface language is almost
never expected to be complete for the core language, or even to have
meaning independently of its elaboration into the core language.
The fulfillment of any such expectation is immediately obstructed by the myriad
non-compositional aspects of the elaboration of surface languages, including
not only the use of unification to resolve implicit arguments and coercions,
but also even the complex name resolution scopes induced by ML's \mathcd{open}
construct.

\subsubsection{Artin gluing and logical relations}

Logical relations, or Tait computability~\citep{tait:1967}, is a method by
which a relation on terms of base type is equipped with a canonical
\emph{hereditary} action on type constructors. The hereditary action can be
seen as a generalization of the induction hypothesis that allows a non-trivial
property of base types to be proved, a perspective summarized in Harper's
tutorial note~\citep{harper:2019:tait}. For instance, let $R\Sub{\TpBool}\subseteq
\mathsf{ClosedTerms}\prn{\TpBool}$ be the property of being either $\ValTt$ or
$\ValFf$; one shows that $R\Sub{\TpBool}$ holds of every closed boolean by lifting
it to each connective in a compositional way:
\[
  f \in R\Sub{\sigma\to \tau} \Leftrightarrow
  \forall x\in R\Sub{\sigma}. f(x) \in R\Sub{\tau}
\]

Other properties (like parametricity) lift to
the other connectives in a similar way. The main obstruction to
replacing this method by a general theorem is the fact that programming
languages are traditionally defined in terms of hand-coded raw terms and operational
semantics; for languages defined in this way, there is \emph{a priori} no way
to factor out the common aspects of logical relations.

In an algebraic setting, however, the syntax of a programming language
is embodied in a particular category equipped with various structures
characterized by universal properties (as detailed in
\cref{sec:algebraic-metatheory}). Here, it is possible to replace the
\emph{method} of logical relations with a \emph{general theory} of logical
relations, namely the theory of Artin gluing.
First developed in the 1970s by the Grothendieck school for the purposes of
algebraic geometry~\citep{sga:4}, Artin gluing can be viewed as a tool to
``stitch together'' a type theory's syntactic category with a category of
semantic things, leading to a category of ``families of semantic things indexed
in syntactic things''. Logical relations are then the proof-irrelevant special
case of gluing, where families are restricted to have subsingleton fibers.

\begin{example}[Canonicity by global sections]
  For instance, let $\CCat$ be the category of contexts and substitutions for a
  given language; the global sections functor $\Mor[\Hom{1}{-}]{\CCat}{\SET}$
  takes each context $\Gamma:\CCat$ to the set $\Hom{1}{\Gamma}$ of closed
  substitutions for $\Gamma$. Then, the gluing of $\CCat$ along $\Hom{1}{-}$
  is the category $\GCat$ of pairs $\prn{\Gamma,\tilde{\Gamma}}$ where
  $\tilde{\Gamma}$ is a family of sets indexed in closing substitutions for
  $\Gamma$; given a closing substitution $\gamma\in\Hom{1}{\Gamma}$, an element
  of the fiber $\tilde\Gamma_\gamma$ should be thought of as evidence
  that $\gamma$ is ``computable''. An object of $\GCat$ is called a
  \emph{computability structure} or a \emph{logical family}.

  The fundamental lemma of logical relations is located in the proof
  that $\GCat$ admits the structure of a model of the given type theory, and
  that the projection functor $\FibMor{\GCat}{\CCat}$ is a homomorphism of
  models. In particular, one may choose to define the $\GCat$-structure of the
  booleans to be the following, letting $q : 2 \to \Hom{1}{\TpBool}$ be the
  function determined by the pair of closed terms $\prn{\ValTt, \ValFf}$:
  \[
    \prn{
      \TpBool,
      \DelimMin{1}
      \brc{i : 2 \mid q(i) = b}\Sub{\prn{b\in \Hom{1}{\TpBool}}}
    }
  \]

  Then, by the fundamental lemma, every closed boolean is either $\ValTt$ or $\ValFf$.
\end{example}

\begin{example}[Binary logical relations on closed terms]
  Rather than gluing along the global sections functor $\Hom{1}{-}$, one may
  glue along $\Hom{1}{-}\times\Hom{1}{-}$: then a computability structure over context
  $\Gamma$ is a family of sets $\tilde{\Gamma}$ indexed in pairs of
  closing substitutions for $\Gamma$. An ordinary binary logical relation is,
  then, a computability structure $\tilde{\Gamma}$ such that each fiber
  $\tilde{\Gamma}\Sub{\gamma,\gamma'}$ is subsingleton.
\end{example}

Because traditional logical relations are defined on raw terms rather than
judgmental equivalence classes thereof, their substantiation requires a great
deal of syntactical bureaucracy and technical lemmas. By working abstractly
over judgmental equivalence classes of typed terms, Artin gluing sweeps away
these inessential details completely, but this is only possible by virtue of
the fact that Artin gluing treats \emph{families} (proof-relevant relations) in
general, rather than only proof-irrelevant relations: the computability of a
given term is a structure with evidence, rather than just a property of the
term.

The proof relevance is important for many applications: for instance, a redex
and its contractum lie in the same judgmental equivalence class, so it would
seem at first that there is no way to treat normalization in a super-equational
way. The insight of \citet{fiore:2002,altenkirch-hofmann-streicher:1995} from
the 1990s is that normal forms can be presented as a \emph{structure} over
equivalence classes of typed terms, rather than as a \emph{property} of raw
terms.  In many cases, the structures end up being fiberwise subsingleton, but
this usually cannot be seen until after the fundamental lemma is proved.

An even more striking use of proof relevance, explained by \citet{shulman:blog:scones-logical-relations,shulman:2015} and \citet{coquand:2019},
is the computability interpretation of
universes. A universe is a special type $\Univ$ whose elements
$A:\Univ$ may be regarded as types $\mathsf{El}\prn{A}\ \mathit{type}$;
in order to substantiate the part of the fundamental lemma that expresses
closure under $\mathsf{El}\prn{-}$, we must have a way to extract a logical
relation over $\mathsf{El}\prn{A}$ from each computable element $A:\Univ$. This would
seem to require a ``relation of relations'', but there can be no such thing:
the fibers of relations are subsingleton.

In the past, some type theorists have accounted for the logical relations of
universes by parameterizing the construction in the graph of an assignment of
logical relations to type codes~\citep{allen:1987:lics}, or by using
induction-recursion; either approach, however, forces the universe to be closed
and inductively defined --- disrupting certain applications of logical
relations, including parametricity. The proof relevance accorded by Artin
gluing offers a more direct solution to the problem: one can always have a
``family of small families''. This insight is also employed in proof-relevant
models of parametricity as discussed in \cref{sec:related:proof-relevant}.

\subsubsection{Synthetic Tait computability for phase separated parametricity}

For a specific type theory, the explicit construction of the gluing
category and the substantiation of the fundamental lemma can be quite
complicated. A major contribution of this paper is a synthetic version
of type-theoretic gluing that situates type theories and their logical relations in the
language of topoi, where we have a wealth of classical results to draw
on~\citep{sga:4,johnstone:2002}: surprisingly, these classical results suffice
to eliminate the explicit and technical constructions of logical
relations and their fundamental lemma, replacing them with elementary
type-theoretic arguments (\cref{sec:parametricity-judgments}).

Following the methodology pioneered (in another context) by
\citet{orton-pitts:2016}, we axiomatize the structure required to work
synthetically with phase separated proof-relevant logical relations
(``parametricity structures''): in \cref{sec:paramtt}, we specify a dependent type
theory \PSTT{} in which every type can be thought of as a parametricity
structure.\footnote{The type theory of synthetic parametricity structures will
turn out to be the internal language of a certain topos $\XTop$, to be defined
in \cref{sec:topos}.} To substantiate the view of \emph{logical relations as
types} we extend \PSTT{} with the following
constructs:

\begin{enumerate}

  \item A proof-irrelevant proposition $\SynOpn$ called the \emph{syntactic
    open} that splits into two disjoint parts $\SynOpn = \LOpn\lor\ROpn$
    corresponding to the left and right components of binary parametricity.
    Then, given a synthetic parametricity structure $A$, we may project the
    \emph{syntactic part} of $A$ as $\MSyn{A} = \prn{\IHom{\SynOpn}{A}}$. It is
    easy to see that $\MSyn$ defines a lex (finite limit preserving) idempotent
    monad, and furthermore commutes with dependent products; a modality defined
    in this way is called an \emph{open modality}.
    Then, a parametricity structure $A$ is called \emph{purely syntactic} if
    the unit $\Mor{A}{\MSyn{A}}$ is an isomorphism.

  \item A proof-irrelevant proposition $\StOpn$ called the
    \emph{static open}; then, given a synthetic parametricity structure $A$,
    the \emph{static part} of $A$ is projected by $\MSt{A} =
    \prn{\IHom{\StOpn}{A}}$, and a \emph{purely static} parametricity structure
    $A$ is one for which $\Mor{A}{\MSt{A}}$ is an isomorphism.

  \item An embedding $\floors{-}$ of \ModTT{}'s syntax as a collection of
    purely syntactic types and functions, such that for any sort $T$ of
    \ModTT{}, the static projection commutes with the embedding: $\floors{\IHom{\StOpn}{T}} \cong \IHom{\StOpn}{\floors{T}}$.

\end{enumerate}

We may then form complementary \emph{closed modalities} $\MSem,\MDyn$ to the
open modalities $\MSyn,\MSt$ that allow one to project the semantic and dynamic
parts respectively of a synthetic parametricity structure, as summarized in
\cref{fig:synthetic-parametricity-structures}. The explanation of their meaning
will have to wait, but we simply note that the ``semantic modality'' $\MSem$ is
the universal way to trivialize the syntactic part of a parametricity
structure, and the ``dynamic modality'' $\MDyn$ is the universal way to
trivialize the static part of a parametricity structure.

\begin{figure}
  \begingroup\small
  \begin{tabular}{llll}
    \textbf{Propositions:} & $\LOpn,\ROpn,\SynOpn,\StOpn:\Omega$ & $\SynOpn=\LOpn\lor\ROpn$ & $\LOpn\land\ROpn=\bot$\\[8pt]
    \textbf{Open modalities:} & $\MSyn,\MSt:\IHom{\Univ}{\Univ}$ & $\MSyn{A}=\IHom{\SynOpn}{A}$ & $\MSt{A} = \IHom{\StOpn}{A}$\\[8pt]
    \textbf{Closed modalities:} & $\MSem,\MDyn:\IHom{\Univ}{\Univ}$ & $\MSem{A}=A\sqcup\Sub{A\times\SynOpn}\SynOpn$ & $\MDyn{A} = A\sqcup\Sub{A\times\StOpn}\StOpn$
  \end{tabular}
  \endgroup

  \caption{A summary of the structure available in the internal language of a
  topos of \emph{synthetic phase separated parametricity structures}. Above,
  $A\sqcup\Sub{A\times\phi}\phi$ is the pushout along the product projections of
  $A\times\phi$. In \cref{sec:closed-modalities,lem:connectedness}, we prove
  the closed modalities are complementary to the open modalities in the sense
  that $\MSyn{\MSem{A}} \cong \ObjTerm{}$ and $\MSt{\MDyn{A}} = \ObjTerm{}$.}
  \label{fig:synthetic-parametricity-structures}
\end{figure}

\paragraph{Synthetic vs.\ analytic Tait computability}

Traditional analytic accounts of Tait computability proceed by defining
exactly how to \emph{construct} a logical relation out of more primitive things
like sets of terms. In contrast, our synthetic viewpoint emphasizes what can be
\emph{done} with a logical relation: the syntactic and semantic parts can be
extracted and pieced together again. The former primitives, such as sets of
terms, then arise as logical relations $A$ such that $A\cong \MSyn{A}$.

Just as Euclidean geometry takes lines and circles as primitives rather than
point-sets, the synthetic account of Tait computability takes the notion of
logical relation as a primitive, characterized by what can be done with
it. Perhaps surprisingly, we have found that all aspects of standard
computability models can be reconstructed in the synthetic setting in a less
technical way.

\NewDocumentCommand\OneML{}{\textsf{1ML}}
\subsection{Discussion of related work}

\subsubsection{\texorpdfstring{\OneML}{1ML} and F-ing Modules}

Most similar in spirit to our module calculus is that of
\OneML{}~\citep{rossberg:2018}, which, as here, uses a universe to represent a
signature of ``small'' types of run-time values. Although \ModTT{} does not
have first class modules, there is no obstacle to supporting the packaging of
modules of small signature into a type. \OneML{} also features a module
connective analogous to the static extent, though the universal property of
this connective is not explicated --- in fact, declarative rules for neither
typing nor equality of modules are specified by
\citet{rossberg-russo-dreyer:2014,rossberg:2018}.
Consequently, the most substantial difference between \ModTT{} and \OneML{} is
that the latter is defined by its translation into System~F$\omega$, whereas
\ModTT{} is given intrinsically as an algebraic theory that expresses equality
of modules, with a modality to confine attention to their static parts.
To be sure, it is elegant and practical to consider the compilation of modules
by a phase-separating translation, as was done for example
by~\citet{petersen:2005}.  Nevertheless, it is also important to give a direct
type-theoretic account of program modules \emph{as they are to be used and
reasoned about}.

It is true that our language too would require elaboration to be usable in
practice, but elaboration here is needed only to introduce subtyping. The
transformation of source code into core language code therefore preserves the
intuitive meanings of all modular constructs. For instance, the meaning of
module hierarchy in our calculus is simply dependent sum; in contrast, the
meaning of a module hierarchy in the F-ing calculi can only be understood by
unraveling the somewhat complex relationship between module signatures and the
$n$-ary iterations of existential types they denote.

A very elegant contribution of \citet{rossberg-russo-dreyer:2014} is a more
compositional reduction of the \emph{dot notation} \cd{M.t} to existential
unpacking than that of \citet{cardelli-leroy:1990}. In light of the fact that
an existential encoding of dependent sums cannot satisfy the $\eta$-law,
however, we find that \citeauthor{macqueen:1986}'s intervention remains in
force today: abstract types do \emph{not} have existential type (\emph{pace}
\citet{mitchell-plotkin:1985,mitchell-plotkin:1988}). We welcome further
exploration of the F-ing interpretation's equational theory, whose relationship
to the equational theory of modules themselves remains somewhat opaque.

\subsubsection{Modules, Abstraction, and Parametric Polymorphism}

In a pair of recent papers~\citep{crary:2017,crary:2019}, Crary develops (1)
the relational metatheory of a calculus of ML modules and (2) a fully abstract
compilation procedure into a version of System~F$\omega$.  Although our two
calculi have similar expressivity, the rules of \ModTT{} are simpler and more
direct; in part, this is because subtyping and retyping are shifted into
elaboration for us, but we also remark that Crary has placed side conditions on
the rules for dependent sums to ensure they only apply in the non-dependent
case, which are unnecessary in \ModTT{}. Crary, however, treats general
recursion at the value level, which we have not attempted in this paper.
In more recent work \citet{crary:2020} joins us in advocating that module projectibility be
reconstructed in terms of a lax modality.

Crary's account of parametricity, the first to rigorously substantiate an
abstraction theorem for modules, achieves a similar goal to our work, but is
much more technically involved.  In particular we have gained much leverage
from working over equivalence classes of typed terms, rather than using
operational semantics on untyped terms --- in fact, our entire development
proceeds without introducing any technical lemmas whatsoever. Another advantage
of our approach is the use of proof relevance to account directly for strong
sums over the collection of types; working in a proof-irrelevant setting, Crary
must resort to an ingenious staging trick in which classes of precandidates are
first defined for every kind, and then the candidates for module signatures are
relations between a pair of module values \emph{and} a precandidate. This can
be seen as a defunctionalization of the proof-relevant interpretation of kinds
by \citet{atkey:2012}, and is not likely to scale to more universes.

\subsubsection{Applicative functor semantics in OCaml}

The interaction between effects and module functors lies at the heart of nearly
all previous work on modules. Leroy proposed an applicative semantics
for module functors~\citep{leroy:1995}, later used in OCaml's module
system~\citep{ocaml:system-manual:4.10.0}: whereas generative functors can be
thought of as functions $\sigma\Rightarrow\SigCmp{\tau}$, applicative functors
correspond roughly to $\SigCmp*{\sigma\Rightarrow\tau}$ as noted by \citet{shao:1999}, but subtleties abound.
The subtleties of applicative and generative functor semantics (studied by
\citet{dreyer-crary-harper:2003} as \emph{weak} and \emph{strong} sealing) are
mostly located in the view of sealing as a computational effect: how can a
structure be ``pure'' if a substructure is sealed? In contrast, we view sealing
in the sense of static information loss as a (clearly pure) projection function
inserted during typechecking, using the user's signature annotations as a
guide. By decoupling sealing from the effect of generating a fresh
abstract type, we obtain a simpler and more type-theoretic account of
generativity embodied in the lax modality.

\subsubsection{Proof-relevant relational interpretation}\label{sec:related:proof-relevant}

We are not the first to consider proof-relevant relational interpretations,
which are essential in the context of dependent type theory because they enable
a compositional interpretation of the universe, an insight employed by
\citet{atkey-ghani-johann:2014,nuyts-vezzosi-devriese:2017}. \citet{atkey:2012}
uses the same insight in his interpretation of \emph{kinds} as reflexive
graphs, with the kind of types given by the reflexive graph of set-theoretic
relations. \citet{sojakova-johann:2018} define a
general framework for parametric models, which can be instantiated to give rise
to a proof-relevant version of parametricity.
\citet{benton-hofmann-nigam:2013,benton-hofmann-nigam:2014} use proof-relevant
logical relations to work around the fact that logical relations involving an
existential quantifier rarely satisfy an important closure condition known as
admissibility, a problem also faced by \citet{crary:2017}.  In the
proof-irrelevant setting this can be resolved either by using continuations
explicitly or by imposing a biorthogonal closure condition that amounts to much
the same thing.

\subsubsection{Syntactic \emph{vs.}\ semantic parametricity}\label{sec:semantic-parametricity}

Parametricity has historically been studied in two forms: the present work is
about \emph{syntactic parametricity} aims to establish identifications between
definable terms within a theory, whereas \emph{semantic parametricity} aims to
establish the compatibility of a theory with certain identifications
by means of a model. Syntactic parametricity can be construed somewhat crudely
as being about \emph{counting} the number of definable elements of a given
type, whereas semantic parametricity has no bearing at all on the question of
how many elements a given type has.
The relationship between the semantic and syntactic parametericity is somewhat
analogous to the difference between a model of typed lambda calculus that
interpets \cd{bool} as $\mathbf{1}+\mathbf{1}$, and a logical relations or gluing
argument that establishes that there are exactly two distinct closed terms of
type \cd{bool}. The former establishes the compatiblity of the language with
the standard booleans, whereas the latter establishes that the language
actually has the standard booleans.

A particularly attractive model of semantic parametricity is given in
\emph{reflexive
graphs}~\citep{robinson-rosolini:1994,atkey:2012,atkey-ghani-johann:2014}. A
reflexive graph is given by an object $E$ of edges, an object $N$ of nodes, two
morphisms $\Mor[\pi_L,\pi_R]{E}{N}$, and a morphism $\Mor[r]{N}{E}$ that is a
section of both $\pi_L,\pi_R$. Reflexive graphs can be seen to be a
proof-relevant generalization of reflexive relations, because the two boundary
projections can be viewed a single map $\Mor{E}{N\times N}$ which need not be a
monomorphism.  As a category of diagrams, reflexive graphs give rise to a
presheaf topos, and hence they have the advantage of being closed under
universes as discussed above (\cref{sec:related:proof-relevant}).

\subsubsection{Internal parametricity}

Abstracting from the semantics of parametricity, it is possible to consider
extensions of dependent type theory for \emph{internal parametricity} that
involve connectives for
relatedness~\citep{krishnaswami-dreyer:2013,bernardy-moulin:2012,bernardy-moulin:2013,bernardy-coquand-moulin:2015,nuyts:2018,cavallo-harper:2020,cavallo:2021}.
Semantic parametricity, especially as embodied in reflexive graphs, can be seen
to be a truncation of a much higher-dimensional structure; going one level up,
one can consider a reflexive graph enriched in reflexive graphs, but there is
no need to stop there. Iterating the reflexive graph construction infinitely,
one gains the ability to speak non-trivially of relatedness of proofs of
relatedness, and so on. Abstracting from these semantics, one obtains
higher-dimensional relatedness connectives as in the work of
\citet{bernardy-coquand-moulin:2015,cavallo-harper:2020}.

As we noted in \cref{sec:semantic-parametricity}, semantic parametricity
differs from syntactic parametricity in that it does not prove any non-trivial
property of a \emph{language}. Internal parametricity can be seen as an extreme
way to resolve this deficiency, by defining a new language that is inspired by
the parametric model.

\subsubsection{Representation independence via univalence}

The \emph{principle of invariance} is a law of structural mathematics stating
that all definable constructs ought to be invariant under
isomorphism~\citep{awodey:2014}. The closure of formal languages for
mathematics under this principle can be seen to be somewhat analogous to
parametricity arguments in which one restricts attention to relations that are
the graphs of isomorphisms, as pointed out by
\citeauthor{martin-lof:nagel:2013} in his Ernest Nagel
Lecture~\citep{martin-lof:nagel:2013}. Voevodsky's \emph{univalence}
principle~\citep{voevodsky:2006,hottbook} is the internalization of this
invariance into a new formal language for mathematics, Homotopy Type Theory /
Univalent Foundations~\citep{hottbook}.

Univalence states that isomorphic types are interchangeable; as a programming
tool, univalence allows one to replace any goal of the form $P(A)$ with one of
the form $P(B)$ provided as one has an isomorphism $A\cong B$. While
univalence is stated only for isomorphisms $A\cong B$ between types, the
principle also applies to structure-preserving isomorphisms between
implementations of \emph{abstract types}, \ie types equipped with operations.
In this way, univalence is an internal parametricity principle for abstract
types \emph{vis-\`a-vis} isomorphisms rather than arbitrary relations.

\citet{acmz:2021} demonstrate that the limitation of univalence to isomorphisms
is not a serious one in practice as far as representation independence is
concerned; on the other hand, actual parametricity is still needed to obtain
\emph{free theorems} in the sense of \citet{wadler:1989}. Most practical
examples of representation independence where the relation is not an
isomorphism can be ``upgraded'' to a structure-preserving isomorphism by
quotienting the representation types on either side. The contribution of
\citet{acmz:2021} is to identify sufficient conditions on a relation required
for such an upgrade to take place, and to develop a library of lemmas and proof
tactics that facilitate the use of univalence to prove internal representation
independence results.

\subsubsection{The Plotkin--Abadi parametricity logic}

Similar in spirit to our efforts is the Plotkin--Abadi logic for parametricity
polymorphism~\citep{plotkin-abadi:1993}, which overlays a logic over System~F
that includes not only the equational theory of System~F but an additional
non-logical axiom scheme for parametricity.  The Plotkin--Abadi logic relies on
a built-in parametricity translation: for each family of System~F types
$\alpha\vdash \sigma\brk{\alpha}$ and relation $R\subseteq \tau_L\times\tau_R$,
there is a relation $\sigma\bbrk{R} \subseteq \sigma\brk{\tau_L}\times
\sigma\brk{\tau_R}$ defined inductively on the structure of $\sigma[\alpha]$.
Then the parametricity axiom scheme asserts for for any polymorphic program $u
: \forall \alpha.\sigma\brk{\alpha}$ and any relation $R\subseteq
\tau_L\times\tau_R$, we have $\prn{u\brk{\tau_L},u\brk{\tau_R}} \in
\sigma\bbrk{R}$.
The resulting logic can be used to refine and prove theorems about System~F
programs that would follow from parametricity. \citet{birkedal-mogelberg:2005}
provide a category-theoretic notion of parametricity that is sound and complete
for the Plotkin--Abadi logic.

Our work can be seen as a more systematic way to recover the consequences of
parametricity. \citet{plotkin-abadi:1993} need to \emph{axiomatize}
parametricity atop a built-in parametricity translation embodied in
the relational instantiation $\sigma\bbrk{R}$. In contrast, we derive
parametricity results from more general considerations; in particular, the
connection between types and relations that lies at the heart of the
Plotkin--Abadi axiom is reflected in our setting by the classic \emph{fracture
theorem} for recollments~\citep{sga:4,rijke-shulman-spitters:2017}. On the
other hand, we do \emph{not} deal with impredicative polymorphism and hence our
approach does not account for System~F; it is plausible that this could be
resolved by replaying our constructions over a realizability topos as in the work of
\citet{pitts:1987,hyland:1988}, but this remains to be verified and is by no means obvious.

\subsubsection{Parametricity translations}

Related to our synthetic account of logical relations, in which the relatedness
of two programs is substantiated by a third program, is the tradition of
parametricity translations exemplified by
\citet{bernardy-jansson-paterson:2012,
pedrot-tabareau-fehrman-tanter:2019,tabareau-tanter-sozeau:2018}, also taken up
by Per Martin-L\"of in his Ernest Nagel Lecture in
2013~\citep{martin-lof:nagel:2013}. In the case of unary parametricity (logical
predicates), one has the \emph{deliverables} translation described by
\citet{mckinna-burstall:1993} and investigated fibrationally by
\citet{hermida:1993}.
Closely related to the binary parametricity translations is the System~R calculus of
\emph{formal parametricity} by \citet{abadi-cardelli-curien:1993}. To put it
somewhat crudely, System~R is a language for programming in the image of a
parametricity translation.

One methodological difference between these and our work is that the
parametricity translations are analytic, explicitly transforming types into
(proof-relevant) logical relations, whereas our theory of parametricity
structures is synthetic: we assume that everything in sight is a logical
relation, and then \emph{modally} isolate the ones that are degenerate in
either the syntactic or semantic direction.

Another significant difference between our work and the parametricity
translations (as well as internal parametricity) is that we account for the
quite common situation in which the semantic parts of parametricity structures
come from an entirely different category than their syntactic parts; this is
important, because there \emph{is} a difference between (e.g.) parametricity
with respect to definable relations, and parametricity with respect to
set-theoretic relations on closed terms. In fact, we make use of this
flexibility in this paper by considering parametricity with respect to
\emph{phase-separated} relations on (phase-separated) closed terms.

\subsubsection{Doubling the syntax}

In \cref{sec:topos} we consider the copower $2\cdot\PrTop{\ThCat}$ of a topos
$\PrTop{\ThCat}$ representing the syntax of \ModTT{}; this ``doubled topos''
serves as a suitable index to a gluing construction, yielding a topos $\XTop =
\prn{\prn{2\cdot\PrTop{\ThCat}} \times \SP} \sqcup\Sub{2\cdot\PrTop{\ThCat}} \SP$ of
phase separated parametricity structures. The fact that doubling the syntax of
a suitable type theory preserves all of its structure was noticed and used effectively by
\citet{wadler:2007}.  This same observation lies at the heart of our convenient
\cref{notation:systems} for working synthetically with the left- and right-hand
sides of parametricity structures.

\subsubsection{Computational effects}

Lax modalities do not interact cleanly with dependent type structure, unlike
the idempotent lex and open modalities of~\citet{rijke-shulman-spitters:2017}.
A potentially promising approach to the integration of real (non-idempotent)
effects into dependent type theory is represented by the \dCBPV{} calculus of
\citet{pedrot-tabareau:2020}, a dependently typed version of Levy's
Call-By-Push-Value~\citep{levy:2004} that treats a hierarchy of universes of
algebras for a given theory in parallel to the ordinary universes of
unstructured types.
We are optimistic about the potential of \dCBPV{} as an improved account of
\emph{certain} effects in dependent type theory, especially those for which the
collection of algebras is itself an algebra. Because not all effects that
we wish to support can have this very strong property, we based our theory on the
more traditional Moggi metalanguage~\citep{moggi:1991}.

At this time, we are unsure how best to account for the addition of general
recursion and general store to our language. For general recursion, one
possibility is to model the semantic parts of parametricity structures in a
topos model of synthetic domain theory~\citep{hyland:1991,fiore-rosolini:1997}
or synthetic guarded domain theory~\citep{bmss:2011}; then the computational
monad might be interpreted as a kind of lifting operation. Higher-order store
is a thornier question whose solution likely involves an application of
synthetic guarded domain theory, but remains elusive; there is a circularity
involved in the semantics of higher-order store that seems, even in the
presence of solutions to the needed domain equations, to preclude a
\emph{proof-relevant} interpretation of the parametricity structures of types.

\section{\texorpdfstring{\ModTT}{ModTT}: a type theory for program modules}

We introduce \ModTT{}, a type-theoretic core language for modules based
on the considerations discussed in \cref{sec:intro}. We first give an informal
description of the language using familiar notations in \cref{sec:informal};
in \cref{sec:lf-encoding}, we present the formal presentation of \ModTT{} in a
logical framework.

\subsection{Informal presentation of \texorpdfstring{\ModTT}{ModTT}}\label{sec:informal}

\subsubsection{Judgmental structure}

\ModTT{} is arranged around three basic syntactic classes: contexts
\fbox{$\IsCx{\Gamma}$}, signatures \fbox{$\IsSig{\Gamma}{\sigma}$}, module
values \fbox{$\IsMod{\Gamma}{V}{\sigma}$}, and module commands
\fbox{$\IsCmp{\Gamma}{M}{\sigma}$}. All judgments presuppose the
well-formedness of their constituents; for readability, we omit many
annotations that \emph{in fact} appear in a formal presentation of \ModTT{};
furthermore, module signatures, values, and commands are all subject to
judgmental equality, and we assume that derivability of all judgments is
closed under judgmental equality. These informal assumptions are substantiated
by the use of a logical framework to give the ``true'' definition of \ModTT{}
in \cref{sec:lf-encoding}.

\subsubsection{Types and dynamic modules}

The simplest module signature is `$\SigTp$', the signature
classifying the object-level types of the programming language, like
$\TpBool$ or $s \rightharpoonup t$. Given a module $t:\SigTp$, there is a
signature $\SigDyn{t}$ classifying the values of the type $t$.
\begin{mathpar}
  \inferrule[type]{
  }{
    \IsSig{\Gamma}{\SigTp}
  }
  \and
  \inferrule[dynamic]{
    \IsMod{\Gamma}{t}{\SigTp}
  }{
    \IsSig{\Gamma}{\SigDyn{t}}
  }
\end{mathpar}

In this section, we do not axiomatize any specific types, though our examples
will require them. This choice reflects our (perhaps heterodox) perspective
that a programming language is a dynamic extension of a theory of
modules, not the other way around.

\subsubsection{Generativity via lax modality}

To reconstruct generativity (\cref{sec:hierarchy-and-functors}) in a type
theoretic way, we employ a modal separation of effects and distinguish commands
(computations) from values. This is achieved by means of a strong monad,
presented judgmentally as a lax modality $\SigCmp$ mediating between the
\fbox{$\IsMod{\Gamma}{V}{\sigma}$} and \fbox{$\IsCmp{\Gamma}{M}{\sigma}$}
judgments.\footnote{A lax modality is exactly the same thing as a strong monad;
at this level, the judgmental distinction between a ``value of signature
$\SigCmp{\sigma}$'' and a ``command of signature $\sigma$'' is blurred, because
one conventionally works up to isomorphism. It would therefore be fine to omit
the form of judgment $\IsCmp{\Gamma}{M}{\sigma}$ from our language, but we keep it for
aesthetic reasons.}
\begin{mathpar}
  \inferrule[formation]{
    \IsSig{\Gamma}{\sigma}
  }{
    \IsSig{\Gamma}{\SigCmp{\sigma}}
  }
  \and
  \inferrule[introduction]{
    \IsCmp{\Gamma}{M}{\sigma}
  }{
    \IsMod{\Gamma}{\ModCmp{M}}{\SigCmp{\sigma}}
  }
  \and
  \inferrule[return]{
    \IsMod{\Gamma}{V}{\sigma}
  }{
    \IsCmp{\Gamma}{\CmpRet{V}}{\sigma}
  }
  \and
  \inferrule[bind]{
    \IsMod{\Gamma}{V}{\SigCmp{\sigma}}\\
    \IsCmp{\Gamma,X:\sigma}{M}{\sigma'}
  }{
    \IsCmp{\Gamma}{\prn{\CmpBind{X}{V}{M}}}{\sigma'}
  }
\end{mathpar}

We also include a reduction rule and a commuting conversion corresponding to the monad laws.

\subsubsection{Module hierarchies and functors}\label{sec:hierarchy-and-functors}

Signatures in \ModTT{} are closed under dependent sum (module hierarchy) and
dependent product (functor), using the standard type-theoretic
rules. We display only the formation rules for brevity:
\begin{mathpar}
  \inferrule[dependent sum]{
    \IsSig{\Gamma}{\sigma}\\
    \IsSig{\Gamma,X:\sigma}{\sigma'}
  }{
    \IsSig{\Gamma}{\SigSg{X}{\sigma}{\sigma'}}
  }
  \and
  \inferrule[dependent product]{
    \IsSig{\Gamma}{\sigma}\\
    \IsSig{\Gamma,X:\sigma}{\sigma'}
  }{
    \IsSig{\Gamma}{\SigPi{X}{\sigma}{\sigma'}}
  }
\end{mathpar}

Generative functors are defined as a mode of use of the dependent product
combined with the lax modality, taking $\prn{\SigPiGen{X}{\sigma}{\sigma'}} :=
\SigPi{X}{\sigma}{\SigCmp{\sigma'}}$ as in \citet{crary:2020}.

\subsubsection{Contexts and the static open}

The usual rules for contexts in Martin-L\"of type theories apply, but we have
an additional context former $\Gamma,\StOpn$ called the \emph{static open}
context:
\begin{mathpar}
  \inferrule{
  }{
    \IsCx{\CxEmp}
  }
  \and
  \inferrule{
    \IsCx{\Gamma}\\
    \IsSig{\Gamma}{\sigma}
  }{
    \IsCx{\Gamma, X : \sigma}
  }
  \and
  \inferrule{
    \IsCx{\Gamma}
  }{
    \IsCx{\Gamma, \StOpn}
  }
\end{mathpar}

\begin{remark}
  The notation is suggestive of the accounts of modal type theory based on dependent right
  adjoints~\citep{clouston-mannaa-mogelberg-pitts-spitters:2018}; indeed, the
  context extension $\prn{-,\StOpn}$ can be seen as a modality on contexts left
  adjoint to a modality on signatures that projects out their static parts.
\end{remark}

The purpose of the static open is to facilitate a context-sensitive version of
judgmental equality in which the dynamic parts of different objects are
identified when $\Gamma\vdash\StOpn$. Specifically, we add rules to ensure that
programs of a given type as well as commands of a given signature are \emph{statically
connected} in the sense of having exactly one element under $\StOpn$, as in
\cref{sec:synthetic-phase-distinction}.
\begin{mathpar}
  \inferrule[static connectivity (1)]{
    \IsMod{\Gamma}{t}{\SigTp}\\
    \Gamma\vdash \StOpn
  }{
    \IsMod{\Gamma}{*}{t}
  }
  \and
  \inferrule[static connectivity (2)]{
    \IsMod{\Gamma}{t}{\SigTp}\\
    \IsMod{\Gamma}{e}{t}\\
    \Gamma\vdash \StOpn
  }{
    \EqMod{\Gamma}{e}{*}{t}
  }
  \and
  \inferrule[static connectivity (3)]{
    \IsSig{\Gamma}{\sigma}\\
    \Gamma\vdash \StOpn
  }{
    \IsCmp{\Gamma}{*}{\sigma}
  }
  \and
  \inferrule[static connectivity (4)]{
    \IsSig{\Gamma}{\sigma}\\
    \IsCmp{\Gamma}{M}{\sigma}\\
    \Gamma\vdash \StOpn
  }{
    \EqCmp{\Gamma}{M}{*}{\sigma}
  }
\end{mathpar}

\subsubsection{The static extent}

The static open is a tool to ensure that dependency is only incurred on the
static parts of objects in \ModTT{}; consequently, we do not include an
equality connective or even a general singleton signature (which would incur a
dynamic dependency). Instead, we introduce the \emph{static extent} of a static
element $\IsMod{\Gamma,\StOpn}{V}{\sigma}$ as the signature
$\SigWhere{\sigma}{V}$ of modules $U:\sigma$ whose static part restricts to $V$;
because our results depend on the algebraic character of \ModTT{}, we
provide explicit introduction and elimination forms for the static extent,
which are trivial to elaborate from an implicit notation.
\begin{mathpar}
  \inferrule[extent/formation]{
    \IsSig{\Gamma}{\sigma}\\
    \IsMod{\Gamma,\StOpn}{V}{\sigma}
  }{
    \IsSig{\Gamma}{\SigWhere{\sigma}{V}}
  }
  \and
  \inferrule[extent/intro]{
    \IsMod{\Gamma}{W}{\sigma}\\
    \EqMod{\Gamma,\StOpn}{W}{V}{\sigma}
  }{
    \IsMod{\Gamma}{\InWhere{V}{W}}{
      \SigWhere{\sigma}{V}
    }
  }
  \and
  \inferrule[extent/elim]{
    \IsMod{\Gamma}{V}{\SigWhere{\sigma}{W}}
  }{
    \IsMod{\Gamma}{\OutWhere{W}{V}}{\sigma}
  }
  \and
  \inferrule[extent/inversion]{
    \Gamma\vdash\StOpn\\
    \IsMod{\Gamma}{V}{\SigWhere{\sigma}{W}}
  }{
    \EqMod{\Gamma}{\OutWhere{W}{V}}{W}{\sigma}
  }
  \and
  \inferrule[extent/$\beta$]{
    \IsMod{\Gamma}{W}{\sigma}\\
    \IsMod{\Gamma,\StOpn}{V}{\sigma}\\
    \EqMod{\Gamma,\StOpn}{W}{V}{\sigma}
  }{
    \EqMod{\Gamma}{\OutWhere{V}{\InWhere{V}{W}}}{W}{\sigma}
  }
  \and
  \inferrule[extent/$\eta$]{
    \IsMod{\Gamma,\StOpn}{W}{\sigma}\\
    \IsMod{\Gamma}{V}{\SigWhere{\sigma}{W}}
  }{
    \EqMod{\Gamma}{V}{
      \InWhere{W}{\OutWhere{W}{V}}
    }{\SigWhere{\sigma}{W}}
  }
\end{mathpar}

The static extent reconstructs both type sharing and weak structure sharing,
which appear in SML '97~\citep{milner-tofte-harper-macqueen:1997} and
OCaml~\citep{ocaml:system-manual:4.10.0}.

\begin{example}
  The SML module signature \textcd{(SHOW where type t = bool)} is rendered in
  terms of the static extent as
  $\SigWhere{\mathcd{SHOW}}{\brk{\mathcd{bool},*}}$, using the \textsc{static
  connectivity (1)} rule from \cref{sec:synthetic-phase-distinction}:
  \begin{prooftree*}
    \infer0{
      \IsMod{\Gamma,\StOpn}{\TpBool}{\SigTp}
    }

    \infer0{
      \Gamma,\StOpn\vdash \StOpn
    }
    \infer1{
      \IsMod{\Gamma,\StOpn}{*}{\SigDyn{\TpBool\rightharpoonup\mathcd{string}}}
    }

    \infer2{
      \IsMod{\Gamma,\StOpn}{\brk{\TpBool,*}}{\mathcd{SHOW}}
    }
    \infer1{
      \IsSig{\Gamma}{\SigWhere{\mathcd{SHOW}}{\brk{\TpBool,*}}}
    }
  \end{prooftree*}
\end{example}

We have (intentionally) made no effort to restrict the families of signatures
to depend only on variables of a static nature, in contrast to previous works
on modules. We conjecture, but do not prove here, the admissibility of a
principle that extends any signature to one that is defined over a purely
static context.
This should follow, roughly, from the fact that genuine dependencies are all
introduced ultimately via the static extent and that there is no signature of
signatures. We note that none of the results of this paper depend on the
validity of this conjecture.

\subsubsection{Further extensions: observables and partial function types}

For brevity, we do not extend \ModTT{} with all the features one would
expect from a programming language. However, our examples will require a type
of observables $\TpBool : \SigTp$ with $\ValTt,\ValFf:\SigDyn{\TpBool}$, as
well as a partial function type $s\rightharpoonup t$ such that
$\SigDyn{s\rightharpoonup t} \cong \SigDyn{s}\Rightarrow\SigCmp{\SigDyn{t}}$.

\subsubsection{External language and elaboration}

We do not present here a surface language, which
would include many features not present in the core language \ModTT{}: for
instance, named fields and paths are elaborated to iterated dependent sum
projections, and SML-style sharing constraints and `\textcd{where type}'
clauses are elaborated to uses of the static extent. Elaboration is essential
to support the implicit dropping and reordering of fields in module signature
matching; furthermore, the crucial subtyping and extensional retyping
principles of \citet{lee-crary-harper:2007} are re-cast as an elaboration
strategy guided by $\eta$-laws, as in the elaboration of extension types in the
\texttt{\textcolor{RegalBlue}{cool}tt} proof assistant~\citep{cooltt:2020}. The
status of subtyping and retyping in \ModTT{} is a significant divergence
from previous work, which treated them within the core language (an
untenable position for an algebraic account of modules).

\subsection{Algebraic presentation in a logical framework}\label{sec:lf-encoding}

Rather than studying directly the informal presentation of \ModTT{} given in
\cref{sec:informal}, we intend to study a mathematical version of this syntax
that can be defined in a logical framework, namely the internal language of
locally Cartesian closed categories. We use the logical framework to capture
not only binding structure, but also well-typedness and judgmental equality.
One important difference between the informal presentation and the logical
framework presentation is that the latter does not distinguish contexts from
other forms of judgment; such a distinction can be important for implementation
as well as establishing various metatheorems (\eg normalization), but it does
not seem to play a role in the specification of the theory itself.

Theories are encoded in the LF (logical framework) as follows:

\begin{enumerate}

  \item Both parameters and hypothetical judgments are formulated using the
    dependent products $(x:X)\to Y(x)$ of the LF.

  \item The LF contains one universe $\JDG$ of \emph{judgments};
    an object-level judgment/sort is defined by adding a constant whose type ends in
    $\JDG$. LF signatures must use $\JDG$ in only strictly positive positions.

  \item Object-level equality is specified by adding constants whose types end
    in the logical framework's equality type $(a =_X b)$.

\end{enumerate}

The LF signature of \ModTT{} is presented in \cref{fig:lf-modtt}.

\begin{figure}
  \begingroup\small
  \begin{align*}
    \StOpn &: \JDG\\
    \_ &: \prn{x,y:\StOpn} \to x =\Sub{\StOpn} y\\
    \Sig &: \JDG\\
    \Val &: \Sig\to\JDG\\
    \SigTp &: \Sig\\
    \SigDyn{-} &: \Val\prn{\SigTp}\to \Sig\\
    \Pi,\Sigma &: \prn{\sigma : \Sig}\to \prn{\Val\prn{\sigma}\to\Sig}\to\Sig\\
    \Con{Ext} &: \prn{\sigma : \Sig} \to \prn{\StOpn\to\Val\prn{\sigma}}\to\Sig\\
    \SigCmp &: \Sig\to\Sig\\
    \Pi\Con{/val} &: \prn{\DelimMin{1}\prn{x:\Val\prn{\sigma}}\to \Val\prn{\tau\prn{x}}} \cong \Val\prn{\Pi\prn{\sigma,\tau}}\\
    \Sigma\Con{/val} &: \prn{\DelimMin{1}\prn{x:\Val\prn{\sigma}}\times \Val\prn{\tau\prn{x}}} \cong \Val\prn{\Sigma\prn{\sigma,\tau}}\\
    \Con{Ext/val} &: \prn{\DelimMin{1}\prn{U:\Val\prn{\sigma}}\times\prn{\prn{z : \StOpn}\to U =\Sub{\Val\prn{\sigma}}V\prn{z}}}\cong \Val\prn{\Con{Ext}\prn{\sigma,V}}\\
    \Cmp &: \Sig\to\JDG\\
    \Cmp &:= \lambda \sigma.\Val\prn{\SigCmp{\sigma}}
    \\
    \Con{conn/dyn} &: \StOpn\to{\Val\prn{\SigDyn{t}}\cong\ObjTerm{}}\\
    \Con{conn/cmp} &: \StOpn\to{\Cmp\prn{\sigma}\cong\ObjTerm{}}\\
    \Con{ret} &: \Val\prn{\sigma}\to \Cmp\prn{\sigma}\\
    \Con{bind} &: \prn{\Val\prn{\sigma}\to\Cmp\prn{\tau}}\to\Cmp\prn{\sigma}\to\Cmp\prn{\tau}\\
    \_ &: \Con{bind}\prn{F,\Con{ret}\prn{V}} =\Sub{\Cmp\prn{\tau}} F\prn{V}\\
    \_ &: \Con{bind}\prn{F,\Con{bind}\prn{G,V}} =\Sub{\Cmp\prn{\rho}} \Con{bind}\prn{\lambda x. \Con{bind}\prn{F,G\prn{x}},V}
    \\
    \prn{\rightharpoonup} &: \Val\prn{\SigTp}\to\Val\prn{\SigTp}\to\Val\prn{\SigTp}\\
    {\rightharpoonup}\Con{/val} &: \prn{\DelimMin{1}\Val\prn{\SigDyn{s}}\to\Cmp\prn{\SigDyn{t}}}\cong\Val\prn{\SigDyn{s\rightharpoonup t}}
    \\
    \TpBool &: \Val\prn{\SigTp}\\
    \ValTt,\ValFf &: \Val\prn{\SigDyn{\TpBool}}\\
    \Con{if} &: \Val\prn{\SigDyn{\TpBool}}\to \Cmp\prn{\SigDyn{t}}\to\Cmp\prn{\SigDyn{t}}\to\Cmp\prn{\SigDyn{t}}\\
    \_ &: \Con{if}\prn{\ValTt,M,N} =\Sub{\Cmp\prn{\SigDyn{t}}} M\\
    \_ &: \Con{if}\prn{\ValFf,M,N} =\Sub{\Cmp\prn{\SigDyn{t}}} N
  \end{align*}
  \endgroup

  \caption{The explicit presentation of \ModTT{} as a signature $\SigModTT$ in
  the logical framework; for readability, we omit quantification over certain
  metavariables.  The introduction, elimination, computation, and uniqueness
  rules of the static extent are captured in a \emph{single} rule
  $\Con{Ext/val}$ declaring an isomorphism, which can be unfolded to a
  dependent sum.}
  \label{fig:lf-modtt}
\end{figure}

\NewDocumentCommand\ALG{mm}{\mathbf{Alg}\Sub{#1}\prn{#2}}

\begin{definition}[Algebras for a signature]\label{def:algebra}
  Let $\Sigma$ be a signature in the LF; the signature $\Sigma$ can be viewed
  as a ``dependent record type'' in any sufficiently structured category
  $\ECat$. In particular, if $\Univ$ is a universe in $\ECat$ closed under
  dependent sum, product, and extensional equality, we have a type
  $\ALG{\Sigma}{\Univ}$ in $\ECat$ defined as the dependent sum of all of the
  components of $\Sigma$ where $\JDG$ is interpreted as $\Univ$; an
  element of $\ALG{\Sigma}{\Univ}$ is then a model of the theory presented by
  $\Sigma$, in which judgments and contexts are $\Univ$-small.
\end{definition}

\paragraph{Syntactic category of an LF signature}

A signature $\Sigma$ in the LF \emph{presents} a certain category
$\CCat_\Sigma$ equipped with all finite limits and some dependent products ---
in the sense that there is a bijection between equivalence classes of LF terms
and morphisms in the category.\footnote{Our application does not, however, rely
on this bijection: instead, we treat the morphisms of the presented category as
a \emph{definition} of our language.} The objects of $\CCat_\Sigma$ are
equivalence classes of judgments over $\Sigma$, and the morphisms are
equivalence classes of deductions.

The notion of an algebra (\cref{def:algebra}) is good for concrete
constructions, but the higher-altitude structure of our development is best
served by \emph{functorial semantics} in the spirit of
Lawvere~\citep{lawvere:thesis}. A model of $\Sigma$ in a sufficiently
structured category $\ECat$ can be viewed in two ways:

\begin{enumerate}
  \item A model of $\Sigma$ is an element of $\ALG{\Sigma}{\Univ}$ for some universe $\Univ$ in a locally Cartesian closed category $\ECat$.
  \item A model of $\Sigma$ is a locally Cartesian closed functor functor $\Mor{\CCat_\Sigma}{\ECat}$.
\end{enumerate}

We will use both perspectives in this paper. The induction principle or
universal property of the syntax states that $\CCat_\Sigma$ is the smallest
model of $\Sigma$; this universal property is the main ingredient for proving
syntactic metatheorems by semantic means, as we advocate and apply in this
paper.

\begin{notation}\label{notation:th-cat}
  We will write $\SigModTT$ for the signature presenting \ModTT{} in
  \cref{fig:lf-modtt}, and $\ThCat$ for the syntactic category of \ModTT{}.
\end{notation}

\paragraph{Equational presentation of specific effects}

It is important that our use of an equational logical framework does not
prevent the extension of \ModTT{} with non-trivial computational effects; although
the effect of having a fixed collection of reference cells or exceptions is
clearly algebraic (see \eg \citet{plotkin-power:2002}), an equational and
structural account of fresh names or nominal restriction is needed in
order to account for languages that feature allocation.

An equational presentation of allocation may be achieved along the lines of
\citet{staton:2013} --- as Staton's work shows, there is no obstacle to
the equational presentation of \emph{any} reasonable form of
deterministic effect, but semantics are another story. We do not currently make any claim
about the extension of our representation independence results to the setting
of higher-order store, for instance.

\section{A type theory for synthetic parametricity}\label{sec:paramtt}

Our goal is to define a ``type theory of parametricity structures'' \PSTT{}, in
which the analytic view of logical relations (as a pair of a syntactic object
together with a relation defined on its elements) is replaced by a streamlined
synthetic perspective, captured under the slogan \textbf{logical relations as
types}. Combined with a model construction detailed in \cref{sec:topos}, the
results of this section will imply a generalized version of the Reynolds
abstraction theorem~\citep{reynolds:1983} for \ModTT{} stated in \cref{thm:reynolds}.

\PSTT{} is an extension of the internal dependent type theory of a presheaf
topos with modal features corresponding to phase separated parametricity:
therefore, \PSTT{} has dependent products, dependent sums, extensional equality
types, a strictly univalent universe $\Omega$ of proof irrelevant propositions,
a strict hierarchy of universes $\Univ{\alpha}$ of types, inductive types,
subset types, and effective quotient types (consequently, strict pushouts). We
first axiomatize \PSTT{} in the style of \citet{orton-pitts:2016}, and in
\cref{sec:topos} we construct a suitable model of \PSTT{} using topos theory.
Referring to the types of \PSTT{}, we will often speak of ``parametricity
structures''.

\subsection{Modal structure of iterated phase separation}
\NewDocumentCommand\MOp{m}{\mathsf{Open}\Sub{#1}}
\NewDocumentCommand\MCl{m}{\mathsf{Closed}\Sub{#1}}

Using the insight that logical relations can be seen as a kind of phase
distinction between the syntactic and the semantic, we iterate the use
of the ``static open'' from \ModTT{} and add to $\PSTT$ a system of proof
irrelevant propositions corresponding to the static part and the
disjoint (left)-syntactic and (right)-syntactic parts of a parametricity structure.
\begin{mathpar}
  \StOpn,\LOpn,\ROpn,\SynOpn :\Omega
  \and
  \LOpn\land\ROpn = \bot
  \and
  \SynOpn := \LOpn\lor\ROpn
\end{mathpar}

\subsubsection{Static and syntactic open modalities}

Using the propositions specified above, we may define \emph{open modalities}
that isolate the static and syntactic aspects of a given type.

\begin{construction}[Open modality]
  If $\phi:\Omega$ is a proposition, then the \emph{open modality}
  corresponding to $\phi$ is $\MOp{\phi}\prn{A} := \phi\to{A}$. One observes
  that the open modality has the following properties:
  \begin{enumerate}

    \item It is monadic: indeed, it is the ``reader monad'' for the proposition $\phi$.

    \item It is idempotent, in the sense that $\MOp{\phi}\prn{\MOp{\phi}{A}}
      \cong \MOp{\phi}\prn{A}$.

    \item It is left exact (``lex'' for short), in the sense that $\MOp{\phi}\prn{a =\Sub{A} b}$ is isomorphic to $\prn{\lambda\_.a}=\Sub{\MOp{\phi}\prn{A}}\prn{\lambda\_.b}$.

    \item It commutes with exponentials, in the sense that $\MOp{\phi}\prn{A\to B}$ is isomorphic to $\MOp{\phi}\prn{A}\to\MOp{\phi}\prn{B}$.

  \end{enumerate}
\end{construction}

\begin{definition}
  When $M$ is an idempotent modality, we say that a type $A$ is
  \emph{$M$-modal} when the unit map $\Mor[\eta]{A}{M\prn{A}}$ is an isomorphism; a
  type $A$ is called \emph{$M$-connected} when $M\prn{A}\cong\ObjTerm{}$.
\end{definition}

We define the ``static modality'' to be $\MSt := \MOp{\StOpn}$ and the
``syntactic modality'' to be $\MSyn := \MOp{\SynOpn}$; the notion of a
$\MOp{\phi}$-modal type gives us an abstract way to speak of types that are
purely syntactic or purely static (or both).

Our open modalities isolate the static and syntactic parts of a parametricity
structure respectively; because $\LOpn,\ROpn$ have no overlap, we have an
isomorphism $\MSyn{A}\cong \prn{\IHom{\LOpn}{A}}\times\prn{\IHom{\ROpn}{A}}$. This
isomorphism is captured more generally by the following \emph{systems} notation
of \citet{cchm:2017} from cubical type theory for constructing maps out of
disjunctions of propositions:

\begin{notation}[Systems]\label{notation:systems}
  Following \citet{cchm:2017}, we employ the notation of \emph{systems} for
  constructing elements of parametricity structures underneath the assumption of disjunction of
  propositions $\phi\lor\phi'$: when $\phi\land\phi'$ implies $a = a':A$, we may write
  $
    \brk{\phi\hookrightarrow a, \phi'\hookrightarrow a'}
  $
  for the unique element of $A$ that restricts to $a,a'$ on $\phi,\phi'$ respectively.
\end{notation}

\begin{notation}[Extension]\label{notation:extension}
  As foreshadowed by the static extents of \ModTT{}, every proposition
  $\phi:\Omega$ gives rise to an \emph{extension type}
  connective~\citep{riehl-shulman:2017}: if $A$ is a parametricity structure
  and $a$ is an element of $A$ assuming $\phi$ is true, then
  $\Compr{A}{\phi\hookrightarrow a}$ is the parametricity structure of elements $a':A$
  such that $a=a'$ when $\phi$ is true.
\end{notation}

\subsubsection{Dynamic and semantic closed modalities}\label{sec:closed-modalities}

The static modality forgets the dynamic part of a parametricity structure
(in both syntax and semantics), and the syntactic modality forgets the
semantic part of a parametricity structure. We will require
complementary modalities to do the opposite, \eg form a parametricity
structure with no syntactic force.

\begin{construction}[Closed modality]\label{con:closed-modality}
  If $\phi:\Omega$ is a proposition, then the \emph{closed} modality
  $\MCl{\phi}$ complementing the open modality $\MOp{\phi} = \prn{\phi\to -}$
  can be defined as a quotient of the product $A\times\phi$ or as a pushout. We
  define $\MCl{\phi}$ in both type theoretic and categorical notation below:

  \medskip
  \begin{minipage}{\textwidth}
    \begin{minipage}{.5\textwidth}
      \begin{code}
        data $\MCl{\phi}$ ($A$ : $\Univ$) where
          $\eta$ : $A\to\MCl{\phi}\prn{A}$
          * : $\phi\to\MCl{\phi}\prn{A}$
          \_ : $\Prod{a:A}{\Prod{z:\phi}{\prn{\eta\prn{a}=*\prn{z}}}}$
      \end{code}
    \end{minipage}
    \begin{minipage}{.5\textwidth}
      \[
        \DiagramSquare{
          se/style = pushout,
          east/style = exists,
          south/style = exists,
          nw = A\times\phi,
          sw = A,
          ne = \phi,
          south = \eta,
          east = *,
          west = \pi_1,
          north = \pi_2,
          se = \MCl{\phi}\prn{A}
        }
      \]
    \end{minipage}
  \end{minipage}

  The modality $\MCl{\phi}$ is lex, idempotent, and monadic, but it does not
  usually commute with exponentials.
\end{construction}

Using \cref{con:closed-modality}, we define the ``purely semantic'' and
``purely dynamic'' modalities respectively:
\begin{align*}
  \MSem &:= \MCl{\SynOpn}\\
  \MDyn &:= \MCl{\StOpn}
\end{align*}

\begin{lemma}\label{lem:connectedness}
  For any $\phi:\Omega$, a type $A$ is $\MCl{\phi}$-modal if and only if it is $\MOp{\phi}$-connected.
\end{lemma}

\begin{proof}
  Suppose that $A$ is $\MCl{\phi}$-modal; to show that $A$ is
  $\MOp{\phi}$-connected, it therefore suffices to show that
  $\MOp{\phi}\prn{\MCl{\phi}\prn{A}}\cong\ObjTerm{}$, which is to say that there is a unique morphism $\Mor{\phi}{\MCl{\phi}\prn{A}}$ given by
  the constructor $\Mor[*]{\phi}{\MCl{\phi}\prn{A}}$. This is clear using the
  induction principle of $\MCl{\phi}\prn{A}$, since the quotienting ensures
  that $\eta\prn{a} = *\prn{z}$ for any $a:A,z:\phi$.

  In the other direction, suppose that $A$ is $\MOp{\phi}$-connected; we must
  check that the unit constructor $\Mor{A}{\MCl{\phi}\prn{A}}$ is an
  isomorphism. We construct the inverse as follows, noting that the
  $\MOp{\phi}$-connectedness of $A$ immediately induces a \emph{unique}
  morphism $\Mor{\phi}{A}$:
  \[
    \begin{tikzpicture}[diagram]
      \SpliceDiagramSquare{
        se/style = pushout,
        nw = A\times\phi,
        sw = A,
        ne = \phi,
        south = \eta,
        east = *,
        west = \pi_1,
        north = \pi_2,
        se = \MCl{\phi}\prn{A},
        width = 2.5cm,
        south/node/style = upright desc,
        east/node/style = upright desc,
      }
      \node (see) [below right = of se] {$A$};
      \path[->,bend right = 30] (sw) edge node [sloped,below] {$\ArrId{A}$} (see);
      \path[->,bend left = 30] (ne) edge (see);
      \path[->,exists] (se) edge (see);
    \end{tikzpicture}
  \]

  We see that $\Mor{\MCl{\phi}\prn{A}}{A}$ is a retraction of the unit, and it
  remains to check that it is a section; this follows immediately from the
  universal property (\ie the $\eta$-law) of the pushout.
\end{proof}

Instantiating \cref{lem:connectedness}, we see the sense in which the pairs of
modalities $\MSyn/\MSem$ and $\MSt/\MDyn$ are each complementary: in
particular, we have $\MSyn{\MSem{A}} \cong \ObjTerm{}$ and $\MSt{\MDyn{A}}
\cong\ObjTerm{}$. Put more crudely, a ``dynamic thing has no static component''
and a ``semantic thing has no syntactic component''. This complementarity is
not the one of boolean logic: the open/closed partition evinces an area of
overlap that is sometimes called the \emph{boundary} or \emph{fringe}, depicted
visually in \cref{fig:geometry}. The geometrical boundary between complementary
the open and closed subspaces is reflected in the modal presentation the fact
that the semantic part of a syntactic thing is \emph{not} trivial, \ie we do
not have $\MSem{\MSyn{A}}\cong \ObjTerm{}$.

\NewDocumentCommand\GenMod{g}{\mathbf{M}\IfValueT{#1}{\prn{#1}}}
\subsection{Universes of modal types}

Each universe $\Univ{\alpha}$ of $\PSTT$ may be restricted to a universe
consisting of \emph{modal} types for each modality described above, \eg a
universe of purely syntactic types or purely dynamic types. Fixing a lex
idempotent modality $\GenMod$, thought to be ranging over
$\brc{\MSyn,\MSem,\MSt,\MDyn}$, we might na\"ively consider defining the
universe $\Univ[\GenMod]{\alpha}$ of $\GenMod$-modal types as a subtype:
\begin{equation*}
  \Univ[\GenMod]{\alpha} :=
  \Compr{
    A : \Univ{\alpha}
  }{
    A\cong\GenMod{A}
  }
  \tag{bad}
\end{equation*}

Unfortunately, such a universe will not itself be $\GenMod$-modal, \ie we do
not have $\GenMod{\Univ[\GenMod]{\alpha}} \cong \Univ[\GenMod]{\alpha}$, hence there
is no hope of closing the $\GenMod$-modal fragment of \PSTT{} under a
hierarchy of universes with such a definition.\footnote{The ``na\"ive''
definition considered here \emph{does} work in homotopy type theories in the
presence of the univalence principle, as shown by
\citet{rijke-shulman-spitters:2017}; because we are working strictly in
ordinary 1-dimensional mathematics, we must choose a different (but
homotopically equivalent) definition of the universe of modal types.} An idea
pioneered in a different context by \citet{streicher:2005} is to apply the
modality directly to the universe:
\begin{equation*}
  \Univ[\GenMod]{\alpha} := \GenMod{\Univ{\alpha}}
  \tag{good}
\end{equation*}

With such a definition, we immediately have $\GenMod{\Univ[\GenMod]{\alpha}}
\cong \Univ[\GenMod]{\alpha}$, \etc; but we still have to specify the decodings
of these new universes, which is to explain what the type of elements of the modal
universe is. This can be done systematically for any modality $\GenMod$, so
long as $\GenMod$ preserves the universe level of types.
Categorically, one views the universe $\Univ{\alpha}$ as a \emph{generic
family} $\Mor[\pi]{\Sum{A:\Univ{\alpha}}{A}}{\Univ{\alpha}}$ that expresses the
indexing of elements over types. The insight of \citet{streicher:2005} was to
apply the modality $\GenMod$ to the entire generic family yielding
$\Mor[\GenMod{\pi}]{\GenMod{\Sum{A:\Univ{\alpha}}{A}}}{\Univ[\GenMod]{\alpha}}$,
and then obtain the collection of elements of a given
$A:\Univ[\GenMod]{\alpha}$ by pullback.

In more type theoretic language, the collection of elements of
$A:\Univ[\GenMod]{\alpha}$ is given by the following decoding map:
\begin{align*}
  \mathsf{El}\Sub{\GenMod} &: \Univ[\GenMod]{\alpha}\to\Univ{\alpha}\\
  \mathsf{El}\Sub{\GenMod}\prn{A} &:=
  \Compr{
    x : \GenMod{\Sum{X:\Univ{\alpha}}{X}}
  }{
    \GenMod{\pi}\prn{x} = A
  }
\end{align*}

We note that each modal universe $\Univ[\GenMod]{\alpha}$ is closed under all the
connectives of \PSTT{}, a general fact about lex idempotent
modalities in topos theory~\citep{maclane-moerdijk:1992} and type
theory~\citep{rijke-shulman-spitters:2017}.

\begin{lemma}\label{fact:el-mod-unit}
  If $A:\Univ{\alpha}$, then $\mathsf{El}\Sub{\GenMod}\prn{\eta\Sub{\GenMod}\prn{A}}\cong \GenMod{A}$.
\end{lemma}

In the case of the open modality for a proposition $\phi:\Omega$, there is a
simpler computation of the decoding of the open subuniverse, namely
$
  \mathsf{El}\Sub{\MOp{\phi}}\prn{A} := \Prod{z:\phi}{A\prn{z}}
$.

\begin{notation}
  From \cref{fact:el-mod-unit}, we are inspired to adopt a slight abuse of
  notation: when $A : \Univ{\alpha}$, we will often write $\GenMod{A} :
  \Univ[\GenMod]{\alpha}$  to mean $\eta\Sub{\GenMod}\prn{A}$; we will also leave
  $\mathsf{El}\Sub{\GenMod}$ implicit, since we have already indulged the
  notational fiction of universes \`a la Russell.
\end{notation}

\subsubsection{Strictification and syntactic realignment}

We assert that the universe hierarchies of \PSTT{} moreover satisfy the
following \emph{strictification} axiom of \citet{orton-pitts:2016,bbcgsv:2016},
which we will justify by a model construction in \cref{sec:topos}.
\begin{axiom}[Strictification]\label{axiom:strictification}
  Let $\phi:\Omega$ be a proposition, and let $A : \phi\to\Univ{\alpha}$ be a
  partial type defined on the extent of $\phi$, and let $B:\Univ{\alpha}$ be a
  total type. Now suppose we have a partial isomorphism $f :
  \Prod{x:\phi}{\prn{A\prn{x}\cong B}}$; then there exists a total type $B'$
  with $g : B'\cong B$, such that both $\forall x : \phi. B' = A\prn{x}$ and
  $\forall x : \phi. f\prn{x} = g$ strictly.
\end{axiom}

\cref{axiom:strictification} above plays a critical role in the constructions
of \cref{sec:modtt-in-pstt}, letting $\phi := \SynOpn$.

\begin{corollary}[Realignment]\label{cor:realignment}
  Let $A : \Univ[\MSyn]{\alpha}$ be a syntactic type, and fix
  $\tilde{A}:\Univ{\alpha}$ whose syntactic part is isomorphic to $A$, i.e. we
  have $f : \MSyn*{A \cong \tilde{A}}$.  Then there exists a type $f^*\tilde{A}:\Univ{\alpha}$
  with $f^\dagger\tilde{A} : f^*\tilde{A}\cong \tilde{A}$, such that both
  $\MSyn*{f^*\tilde{A} = A}$ and $\MSyn*{f^\dagger\tilde{A} = f}$ strictly.
\end{corollary}

\subsection{Doubled embedding of syntax}\label{sec:doubled-embedding}

We need to embed the syntax of \ModTT{} into the syntactic fragment of \PSTT{}.
This is done by assuming a $\SigModTT$-algebra valued in a universe
$\Univ[\MSyn]$ of purely syntactic types, \ie an element
$\mathcal{A}\SubSyn : \ALG{\SigModTT}{\Univ[\MSyn]}$. Because we have
specified $\SynOpn = \LOpn\lor\ROpn$, we also obtain ``left-syntactic'' and
``right-syntactic'' algebras
$\mathcal{A}\SubL,\mathcal{A}\SubR$ respectively such that
$
  \mathcal{A}\SubSyn = \brk{
    \LOpn\hookrightarrow\mathcal{A}\SubL,
    \ROpn\hookrightarrow\mathcal{A}\SubR
  }
$.

\begin{notation}[Syntactic embedding]
  The algebra $\mathcal{A}\SubSyn:\ALG{\SigModTT}{\Univ[\MSyn]}$
  determines, by projection, an object corresponding to each piece of syntax
  definable in \ModTT{}. For instance, the object of \ModTT{}-signatures is
  obtained by the projection $\mathcal{A}\SubSyn.\Sig:\Univ[\MSyn]$. To
  lighten the notation we will write these projections informally as $\EmbLR{\Sig}$,
  \etc, writing $\EmbL{\Sig},\EmbR{\Sig}$ for the corresponding projections
  from the induced left-syntactic and right-syntactic algebras respectively.
\end{notation}

To complete our axiomatization of the embedding of \ModTT{} into \PSTT,
we additionally require that under the assumption of $\SynOpn$, we have $\StOpn
= \EmbLR{\StOpn}$; in other words, we require $\StOpn :
\Compr{\Omega}{\SynOpn\hookrightarrow \EmbLR{\StOpn}}$.

\subsection{A parametric model of \texorpdfstring{\ModTT}{ModTT} in \texorpdfstring{\PSTT}{TT/PS}}\label{sec:modtt-in-pstt}

In this section, we exhibit a second algebra for \ModTT{} in \PSTT{} that lies
over the doubled embedding described in \cref{sec:doubled-embedding}.  To be
precise, we will construct an algebra with the following ``syntactic
extent'' type for some sufficiently large universe $\Univ$:
\[
  \mathcal{A} : \Compr{
    \ALG{\SigModTT}{\Univ}
  }{
    \SynOpn\hookrightarrow \mathcal{A}\SubSyn
  }
\]

We do not show every part of the construction of this ``parametric algebra'',
but instead give several representative cases to illustrate the comparative
ease of our approach in contrast to prior work on proof relevant logical
relations~\citep{sterling-angiuli-gratzer:2019,sterling-angiuli:2020,coquand:2019,kaposi-huber-sattler:2019}
and conventional logical
relations~\citep{crary:2017,gratzer-sterling-birkedal:2019,angiuli:2019} for
dependent types.

\subsubsection{Parametricity structure of judgments}\label{sec:parametricity-judgments}

We define a parametricity structure of signatures over
the purely syntactic parametricity structure of syntactic signatures
$\EmbLR{\Sig}$. Letting $\alpha<\beta<\gamma$, we define $\Sig:
\Univ{\beta}$ with the following interface:
\begin{align*}
  \Sig&: \Compr{\Univ{\gamma}}{\SynOpn\hookrightarrow \EmbLR{\Sig}}\\
  \Sig&\cong
  \Sum{
    \sigma : \EmbLR{\Sig}
  }{
    \Compr{
      \Univ{\beta}
    }{
      \SynOpn \hookrightarrow \EmbLR{\Val}\prn{\sigma}
    }
  }
\end{align*}

The construction of $\Sig$ proceeds in the following way. First, we
define $\Sig'$ to be the dependent sum
$\Sum{ \sigma : \EmbLR{\Sig} }{ \Compr{ \Univ{\beta} }{ \SynOpn \hookrightarrow
\EmbLR{\Val}\prn{\sigma} } }$.
We observe that there is a canonical partial isomorphism $f :
\MSyn{\Sig' \cong \EmbLR{\Sig}}$; supposing $\SynOpn=\top$, it suffices
to construct an ordinary isomorphism:
\begin{align*}
  \Sig' &=
  \Sum{
    \sigma : \EmbLR{\Sig}
  }{
    \Compr{\Univ{\beta}}{
      \SynOpn \hookrightarrow \EmbLR{\Val}\prn{\sigma}
    }
  }
  &&\text{def.\ of $\Sig'$}
  \\
  &=
  \Sum{
    \sigma : \EmbLR{\Sig}
  }{
    \Compr{\Univ{\beta}}{
      \top \hookrightarrow \EmbLR{\Val}\prn{\sigma}
    }
  }
  &&\text{$\SynOpn=\top$}
  \\
  &\cong
  \Sum{
    \sigma : \EmbLR{\Sig}
  }{
    \ObjTerm{}
  }
  &&\text{singleton}
  \\
  &\cong \EmbLR{\Sig}
  &&\text{trivial}
\end{align*}

Therefore, by \cref{cor:realignment} we
obtain $\Sig\cong \Sig'$ strictly extending
$\EmbLR{\Sig}$ as desired.
Next, we may define the collection of elements of a glued signature directly:
\begin{align*}
  \Val&: \Compr{\IHom{\Sig}{\Univ{\beta}}}{
    \SynOpn\hookrightarrow \EmbLR{\Val}
  }
  \\
  \Val\prn{\sigma,\tilde\sigma} &= \tilde\sigma
\end{align*}

\subsubsection{Parametricity structure of dependent products}
We show that $\Sig$ is closed under dependent product (dependent sums are analogous); fixing $\sigma_0
: \Sig$ and $\sigma_1 : \Val\prn{\sigma_0}
\to\Sig$, we may define
$\Pi\Sub{\Sig}\prn{\sigma_0,\sigma_1} : \Sig$ as follows. We desire
the first component to be the syntactic dependent product type
$
  \sigma_\Pi =
  \EmbLR{\Pi\Sub{\Sig}}\prn{\sigma_0,\lambda x:\EmbLR{\Val}\prn{\sigma_0}. \sigma_1\prn{x}}
$.\footnote{We note that we always have $\SynOpn=\top$ in scope when constructing an element of $\EmbLR{\Sig}$.}
For the second component, we note that the syntactic modality commutes with
dependent products up to isomorphism, so (using \cref{cor:realignment}) we may
define the second component lying strictly over $\sigma_\Pi$:
\begin{align*}
  \tilde{\sigma}_\Pi &: \Compr{\Univ{\beta}}{\SynOpn\hookrightarrow \EmbLR{\Val}\prn{\sigma_\Pi}}
  \\
  \tilde{\sigma}_\Pi &\cong \Pi\Sub{\Univ{\beta}}\prn{
    \Val\prn{\sigma_0},
    \Val\circ\sigma_1
  }
\end{align*}

Because we used the dependent product $\Pi\Sub{\Univ{\beta}}$ of \PSTT{}, we
automatically have an appropriate model of the $\lambda$-abstraction,
application, computation, and uniqueness rules without further work.

\begin{remark}
  The parametricity structure of the dependent product is the ``proof'' that
  our synthetic approach is a big step forward (\eg compared to the explicit
  constructions of \citet{kaposi-huber-sattler:2019,sterling-angiuli:2020}). In
  those formulations one constantly uses the fact that the gluing functor
  preserves finite limits, and it is non-trivial to show that the resulting
  construction is in fact a dependent product (which is here made trivial). The
  work did not disappear: it is in fact located in several pages of SGA~4, in
  which certain comma categories are proved to satisfy the Giraud axioms of a
  category of sheaves~\citep{sga:4}, a result that is easier to prove in
  generality than any specific type theoretic corollary.\footnote{Later on we
  simplify matters further by making use of the closure of presheaf topoi under
  gluing along continuous functors.}
\end{remark}

\subsubsection{Parametricity structure of types}

From the syntax of \ModTT{}, we have the signature of types $\EmbLR{\SigTp} : \EmbLR{\Sig}$ and its decoding
$\SigDyn{-}\SubSyn : \EmbLR{\Val\prn{\SigTp}\to\Sig}$; we must provide parametricity
structures for both. First, we may define a collection of small statically connected
parametricity structures for types, using \cref{cor:realignment}:
\begin{align*}
  \mathbf{Type}&: \Compr{\Univ{\beta}}{\SynOpn\hookrightarrow \EmbLR{\Val\prn{\SigTp}}}\\
  \mathbf{Type}&\cong
  \Sum{
    t : \EmbLR{\Val\prn{\SigTp}}
  }{
    \Compr{\Univ[\MDyn]{\alpha}}{
      \SynOpn \hookrightarrow \EmbLR{\Val}\SigDyn{t}\SubSyn
    }
  }
\end{align*}

We may therefore construct the parametricity structure of the signature of types:
\begin{align*}
  \SigTp&: \Compr{\Sig}{\SynOpn\hookrightarrow \EmbLR{\SigTp}}
  &
  \SigDyn{-}&: \Compr{\Val\prn{\SigTp} \to \Sig}{\SynOpn\hookrightarrow \SigDyn{-}\SubSyn}
  \\
  \SigTp&= \prn{\EmbLR{\SigTp}, \mathbf{Type}}
  &
  \SigDyn{\prn{t,\tilde{t}}}&= \prn{\SigDyn{t}\SubSyn, \tilde{t}}
\end{align*}

\subsubsection{Parametricity structure of observables}\label{sec:parametricity-booleans}

We have a type $\TpBool : \EmbLR{\Val}\prn{\EmbLR{\SigTp}}$ and two constants $\ValTt,\ValFf :
\EmbLR{\Val}\prn{\SigDyn{\EmbLR{\TpBool}}}$; we must construct parametricity
structures for all these. First, we define the collection of computable
booleans as follows, using \cref{cor:realignment} as usual:\footnote{Observe that the second component of the dependent sum is a singleton when $\SynOpn = \top$.}
\begin{align*}
  \mathbf{bool}&: \Compr{\Univ{\alpha}}{\SynOpn\hookrightarrow\EmbLR{\Val\prn{\SigDyn{\TpBool}}}}\\
  \mathbf{bool}&\cong
  \Sum{
    b : \EmbLR{\Val\prn{\SigDyn{\TpBool}}}
  }{
    \MDyn
    \MSem
    \brc{
      \tilde{b} : 2
      \mid
      b = \mathsf{case}\brk{\tilde{b}}\prn{\EmbLR{\ValTt},\EmbLR{\ValFf}}
    }
  }
\end{align*}

The application of the closed modality $\MDyn$ ensures that
the values of observable type have no static part (they are ``statically
connected''). We may therefore define the type of booleans:
\begin{align*}
  \TpBool&: \Compr{\Val\prn{\SigTp}}{\SynOpn\hookrightarrow \EmbLR{\TpBool}}
  \\
  \TpBool&= \prn{\EmbLR{\TpBool}, \mathbf{bool}}
\end{align*}

The parametricity structures for the observable values are defined as follows:
\begin{gather*}
  \ValTt,\ValFf: \Compr{\Val\prn{\SigDyn{\TpBool}}}{\SynOpn\hookrightarrow\EmbLR{\ValTt},\EmbLR{\ValFf}}\\
  \ValTt= \prn{
    \EmbLR{\ValTt},
    \DelimMin{1}
    \upeta\Sub{\MDyn}\prn{
      \upeta\Sub{\MSem}\prn{
        0
      }
    }
  }
  \qquad
  \ValFf= \prn{
    \EmbLR{\ValFf},
    \DelimMin{1}
    \upeta\Sub{\MDyn}\prn{
      \upeta\Sub{\MSem}\prn{
        1
      }
    }
  }
\end{gather*}

\subsubsection{Parametricity structure of computational effects}

In this section, we show how to construct a monad on parametricity structures
corresponding to the lax modality of \ModTT{}, following an internal version of
the recipe of \citet{goubault-larrecq-lasota-nowak:2008} for gluing together
two monads along a monad morphism.
Emanating from the syntax is an internal monad $\EmbLR{\SigCmp} :
\EmbLR{\Sig}\to\EmbLR{\Sig}$ on the internal category of syntactic signatures;
here we describe how to glue this monad together with a monad on
the internal category of purely semantic parametricity structures. Let
$\MonadT : \Univ[\MSem]{\beta}\to\Univ[\MSem]{\beta}$ be such a monad;
we furthermore have an internal functor $F : \EmbLR{\Sig}\to
{\Univ[\MSem]{\beta}}$ defined by taking the purely semantic part of the
collection of modules of every syntactic signature:
\begin{align*}
  F\prn{\sigma} &= \MSem{\EmbLR{\Val}\prn{\sigma}}
\end{align*}

\NewDocumentCommand\MonadRun{og}{\mathsf{run}\IfValueT{#1}{\Sub{#1}}\IfValueT{#2}{\prn{#2}}}

We parameterize the constructions of this section in a monad morphism
$\MonadRun : \EmbLR{\SigCmp}\to \MonadT$ over $F$ in the sense of
\citet{street:1972}, \ie an internal natural transformation $\MonadRun :
\MonadT\circ F \to F \circ \EmbLR{\SigCmp}$ satisfying a number of coherence
conditions. Following \citet{goubault-larrecq-lasota-nowak:2008}, we may glue
the two monads together along this morphism to define a monad on
$\Sig$, \ie the internal category of glued signatures and glued
modules determined by the constructions in \cref{sec:parametricity-judgments}.
Fixing $\sigma : \Sig$, we may define a type $\MonadT<\XTop>{\sigma} :
\Univ{\beta}$ as follows, writing $\pi_\sigma :
\MSem{\Val\prn{\sigma}}\to\MonadT{F\prn{\sigma}}$ for the induced projection in $\Univ[\MSem]{\beta}$:
\begin{align*}
  \MonadT<\XTop> &: \Prod{\sigma:\Sig}{
    \Compr{\Univ{\beta}}{
      \SynOpn\hookrightarrow\EmbLR{\Val}\prn{\EmbLR{\SigCmp}{\sigma}}
    }
  }
  \\
  \MonadT<\XTop>{\sigma} &\cong
  \Sum{
    x^\PtO : \EmbLR{\SigCmp}{\sigma}
  }{
    \Compr{
      \DelimMin{1}
      x^\PtC : \MonadT{\MSem{\Val\prn{\sigma}}}
    }{
      \MonadRun[\sigma]{
        \MonadT{\pi_\sigma}\prn{x^\PtC}
      }
      =
      \upeta\Sub{\MSem}\prn{x^\PtO}
    }
  }
\end{align*}

Therefore, we may define the monad on parametricity
structures for signatures as follows:
\begin{align*}
  \SigCmp&: \Sig\to \Sig\\
  \SigCmp{\sigma} &=
  \prn{
    \EmbLR{\SigCmp}{\sigma},
    \MonadT<\XTop>\prn{\sigma}
  }
\end{align*}

If $\ModTT{}$ is suitably extended by monadic operations (such as those
corresponding to exceptions, printing, a global reference cell, \etc), then the
assumptions of this section are readily substantiated by the corresponding
monad on purely semantic objects.
Some computational effects may require the constructions of \cref{sec:topos} to
be relativized from $\SET$ to a suitable presheaf category--- for instance,
partiality / general recursion might be modeled by replacing $\SET$ with the
topos of trees as in \citet{bmss:2011,paviotti:2016} (but we do not make any
claims in this direction).

\begin{example}\label{ex:crash}
  Suppose that $\ModTT{}$ were extended with an operation
  $\mathsf{throw} : \SigCmp{\sigma}$ for each signature $\sigma$, such
  that $\SigCmp$ corresponds to the exception monad. We may glue this
  together with the internal monad $\MonadT{X} = \MDyn{1+X}$ on the internal
  category of purely semantic parametricity structures. We must define a family
  of functions $\MonadRun[\sigma] : \MonadT{F{\prn{\sigma}}}\to
  F\prn{\EmbLR{\SigCmp}{\sigma}}$. Because $F\prn{\EmbLR{\SigCmp}{\sigma}}$ is purely dynamic and
  $\MDyn$ is a lex idempotent modality, any such function $\MonadRun[\sigma]$ is uniquely
  determined by a map $1+F\prn{\sigma}\to
  F\prn{\sigma}$, which we may choose as follows:
  \begin{mathpar}
    \mathit{inl}\prn{*} \mapsto \upeta\Sub{\MSem}\prn{\EmbLR{\mathsf{throw}}}
    \and
    \mathit{inr}\prn{x} \mapsto \upeta\Sub{\MSem}\prn{\EmbLR{\CmpRet}\prn{x}}
  \end{mathpar}

  Then, the monad $\MonadT<\XTop>{\sigma}$ on a parametricity structure
  $\sigma:\Sig$ associates to each syntactic computation
  $M:\EmbLR{\SigCmp}{\sigma}$ either a proof that $M$ throws the exception or a proof that $M$
  returns a computable value.
\end{example}

\section{Case study: representation independence for queues}

In this section, we consider an extension of \ModTT{} by an inductive type of
lists, as well as the \textcd{throw} effect of \cref{ex:crash}. For the purpose
of readability, we adopt a high-level notation for modules and their signatures
where components are identified by name rather than by position.

\subsection{A simulation structure between two queues}

\begin{figure}
  \begin{code}
signature QUEUE =
sig
  type t
  val emp : t
  val ins : bool * t $\rightharpoonup$ t
  val rem : t $\rightharpoonup$ bool * t
end
\-
structure ListQueue : QUEUE =
struct
  type t = bool list
  val emp = nil
  fun ins (x, q) = ret (x :: q)
  fun rem q =
    bind val rev\_q $\leftarrow$ rev q in
    case rev\_q of
    | nil $\Rightarrow$ throw
    | x :: xs $\Rightarrow$
      bind val rev\_xs $\leftarrow$ rev xs in
      ret (f, rev\_xs)
end
\-
structure BatchedQueue : QUEUE =
struct
  type t = bool list * bool list
  val emp = (nil, nil)
  fun ins (x, (fs, rs)) = ret (fs, x :: rs)
  fun rem (fs, rs) =
    case fs of
    | nil $\Rightarrow$
      bind val rev\_rs $\leftarrow$ rev rs in
      (case rev\_rs of
       | nil $\Rightarrow$ throw
       | x::rs' $\Rightarrow$ ret (x, rs', nil))
    | x::fs' $\Rightarrow$ ret (x, fs', rs)
end
  \end{code}
  \caption{Two implementations of a queue in an extended version of \ModTT{}, written in SML-style notation.}
  \label{fig:queues}
\end{figure}

We may define an abstract type of queues $\EmbLR{\mathcd{QUEUE}}$ together with two implementations as
in \citet{harper:2016}, depicted in \cref{fig:queues}.
We will observe that the semantic part of
$\mathcd{QUEUE}$ is the collection of
proof-relevant phase separated simulation relations between two given closed
syntactic queues. First, we note the meaning of $\mathcd{QUEUE}$ in the glued algebra:
\begin{align*}
  \mathcd{QUEUE} &\cong
  \Sum{t:\SigTp}{
    \SigDyn{t}
    \times
    \SigDyn{\TpBool * t\rightharpoonup t}
    \times
    \SigDyn{t\rightharpoonup{\TpBool* t}}
  }
\end{align*}

\NewDocumentCommand\Bits{}{\Con{bits}}

The two implementations internalize as elements $\cd{ListQueue} :
\EmbL{\Val\prn{\mathcd{QUEUE}}}, \cd{BatchedQueue} :
\EmbR{\Val\prn{\mathcd{QUEUE}}}$; these can be combined into
$\cd{Q}:\EmbLR{\Val\prn{\mathcd{QUEUE}}}$ by splitting:
\[
  \cd{Q} = \brk{\LOpn\hookrightarrow \cd{ListQueue},
\ROpn\hookrightarrow \cd{BatchedQueue}}
\]

We may define a purely dynamic type that represents the invariant structure on
a pair of queues using \cref{cor:realignment}, writing $\Bits{} = 2^\star$ for the \PSTT{}-type
of finite lists of bits and $\ceils{-}$ for the obvious
projection of a syntactic element of \ModTT{}-type $\Con{list}\prn{\Con{bool}}$ from finite list of bits.
\[
  \begin{array}{l}
    \Con{invariant} :
    \Compr{\Univ[\MDyn]{\alpha}}{
      \SynOpn \hookrightarrow \MDyn\relax\MSyn\relax {\Val\prn{\cd{Q}.\mathcd{t}}}
    }
    \\
    \Con{invariant} \cong\\
    \ \
    \begin{array}{l}
      \Sum{
        q : \EmbLR{\Val}\prn{\SigDyn{\cd{Q}.\mathcd{t}}}
      }\\
      \MSem
      \Compr{
        \vec{x},\vec{y},\vec{z} : \MDyn\Bits
      }{
        \vec{x} = \prn{\vec{y}+\mathit{rev}\prn{\vec{z}}}
        \land
        q = \brk{\LOpn \hookrightarrow \ceils{\vec{x}}\mid \ROpn\hookrightarrow \prn{\ceils{\vec{y}},\ceils{\vec{z}}}}
      }
    \end{array}
  \end{array}
\]

We may then define a single parametricity structure to unite the two
implementations under the invariant above, depicted in \cref{fig:simulation};
it is now possible to prove the central result of our case study, the representation independence theorem for queues.

\begin{figure}
  \raggedright
  A simulation over $\cd{Q} = \brk{\LOpn\hookrightarrow\cd{ListQueue}, \ROpn\hookrightarrow\cd{BatchedQueue}}$ consists of the following data:
  \begin{align*}
    t &: \Compr{\Val\prn{\SigTp}}{\SynOpn\hookrightarrow \cd{Q}.\mathcd{t}}
    \\
    \mathit{emp} &: \Compr{\Val\prn{\SigDyn{t}}}{\SynOpn\hookrightarrow \cd{Q}.\mathcd{emp}}
    \\
    \mathit{ins} &: \Compr{\Val\prn{\SigDyn{\TpBool*t\rightharpoonup t}}}{\SynOpn\hookrightarrow\cd{Q}.\mathcd{ins}}
    \\
    \mathit{rem} &: \Compr{
      \Val\prn{
        \SigDyn{t\rightharpoonup{\TpBool* t}}
      }
    }{\SynOpn\hookrightarrow\cd{Q}.\mathcd{rem}}
  \end{align*}

  \medskip

  These operations are implemented in \PSTT{} as follows.
  \begingroup
  \scriptsize
  \begin{align*}
    t &= \prn{\cd{Q}.\mathcd{t},\mathsf{invariant}}\\
    \mathit{emp} &= \prn{\cd{Q}.\mathcd{emp}, \prn{\gl{},\gl{},\gl{}}}\\
    \mathit{ins}\prn{\prn{b,x},\prn{q,\prn{\vec{x},\vec{y},\vec{z}}}} &=
    \prn{
      \cd{Q}.\mathcd{ins}\prn{b,q},
      \upeta\Sub{\MonadT}\prn{
        \brk{
          \LOpn\hookrightarrow b :: q,
          \ROpn\hookrightarrow{\prn{\mathcd{fst}\prn{q},b::\mathcd{snd}\prn{q}}}
        },
        \prn{x :: \vec{x}, x::\vec{y}, \vec{z}}
      }
    }
    \\
    \mathit{rem}\prn{q,\prn{\vec{x},\vec{y},\vec{z}}}.1 &= \cd{Q}.\mathcd{rem}\prn{q}\\
    \mathit{rem}\prn{q,\prn{\gl{},\gl{},\gl{}}}.2 &=
      \mathit{throw}\Sub{\MonadT}
    \\
    \mathit{rem}\prn{q,\prn{\vec{x}\ldots x}, x :: \vec{y},\vec{z}}.2
    &=
      \upeta\Sub{\MonadT}%
      (%
        (%
          \ceils{x},
          [%
            \LOpn\hookrightarrow \ceils{\vec{x}}
            \mid
            \ROpn\hookrightarrow (\ceils{\vec{y}},\ceils{\vec{z}})%
          ]%
        ),%
        (%
          x, (\vec{x},\vec{y},\vec{z})%
        )%
      )%
    \\
    \mathit{rem}(q,((\vec{x}\ldots x), \gl{}, \vec{z}\ldots x)).2
    &=
      \upeta\Sub{\MonadT}
      \prn{
        (%
          \ceils{x},
          [%
            \LOpn\hookrightarrow \ceils{\vec{x}}
            \mid
            \ROpn\hookrightarrow (\ceils{\mathit{rev}(\vec{z})},\ceils{\vec{y}})%
          ]%
        ),%
        (x, (\vec{x},\mathit{rev}(\vec{z}),\vec{y}))
      }
  \end{align*}

  \medskip
  where
  \begin{align*}
    \mathsf{invariant} &: \Compr{\Univ[\MDyn]{\alpha}}{
      \SynOpn \hookrightarrow \MDyn\relax\MSyn\relax {\Val\prn{\cd{Q}.\mathcd{t}}}
    }
    \\
    \mathsf{invariant} &\cong
    \Sum{
      q : \MSyn\relax {\Val\prn{\SigDyn{\cd{Q}.\mathcd{t}}}}
    }{
      \MSem
      \Compr{
        \vec{x},\vec{y},\vec{z} : \MDyn\Bits
      }{
        \vec{x} = \prn{\vec{y}+\mathit{rev}\prn{\vec{z}}}
        \land
        q = \brk{\LOpn \hookrightarrow \ceils{\vec{x}}\mid \ROpn\hookrightarrow \prn{\ceils{\vec{y}},\ceils{\vec{z}}}}
      }
    }
  \end{align*}
  \endgroup

  \caption{Constructing a simulation between the two queue implementations
  becomes a straightforward \emph{programming problem} in \PSTT.}

  \label{fig:simulation}
\end{figure}

\begin{theorem}\label{thm:repr-ind}
  Let $f : \EmbLR{\mathcd{QUEUE}\to \SigDyn{\TpBool}}$; then we have $f\prn{\cd{ListQueue}} =
  f\prn{\cd{BatchedQueue}}$.
\end{theorem}
\begin{proof}
  This can be seen by considering the image of $f$ under the
  parametricity interpretation of $\ModTT{}$ into \PSTT, $\tilde{f} :
  \mathcd{QUEUE}\to\SigDyn{\TpBool}$. Applying $\tilde{f}$ to the
  simulation queue defined in \cref{fig:simulation}, we have a single element of
  $\SigDyn{\TpBool}$ relating two syntactic booleans:
  \[
    b : \Compr{\SigDyn{\TpBool}}{\LOpn\hookrightarrow \EmbL{f\prn{\cd{ListQueue}}}, \ROpn\hookrightarrow\EmbR{f\prn{\cd{BatchedQueue}}}}
  \]

  But we have defined $\TpBool$ along the diagonal
  (\cref{sec:parametricity-booleans}), so this actually proves that either
  $f\prn{\cd{ListQueue}} = f\prn{\cd{BatchedQueue}} = \ValTt$ or $f\prn{\cd{ListQueue}} = f\prn{\cd{BatchedQueue}} = \ValFf$.
\end{proof}

\section{The topos of phase separated parametricity structures}\label{sec:topos}

The simplest way to substantiate the type theory \PSTT{} of \cref{sec:paramtt} is
to use the existing infrastructure of Grothendieck topoi and Artin
gluing~\citep{sga:4}; every topos possesses an extremely rich \emph{internal
type theory}, so our strategy will be roughly as follows:

\begin{enumerate}

  \item Embed the syntax of \ModTT{} into a topos $\PrTop{\ThCat}$; this will
    be the topos corresponding to the free cocompletion of the syntactic
    category $\ThCat$ (see \cref{notation:th-cat}). The copower $2\cdot\PrTop{\ThCat}$ will then serve as a
    suitable index for binary parametricity.

  \item Identify a topos $\SP$ that captures the notion of phase distinction:
    a type in the internal language of $\SP$ should be a set that
    has both a static part and a dynamic part depending on it.

  \item Glue the topos of (doubled) syntax $2\cdot\PrTop{\ThCat}$ and the topos
    of semantics $\SP$ together to form a topos $\XTop$ of \emph{phase
    separated parametricity structures}: a type in the internal language of
    $\XTop$ will have several aspects corresponding to the orthogonal
    distinctions ((left syntax, right syntax), semantics) and (static,
    dynamic). The topos $\XTop$ then has enough structure to model all of
    \PSTT{}.

\end{enumerate}

\begin{figure}[h!]
  \begin{center}
    \includegraphics[width=3in]{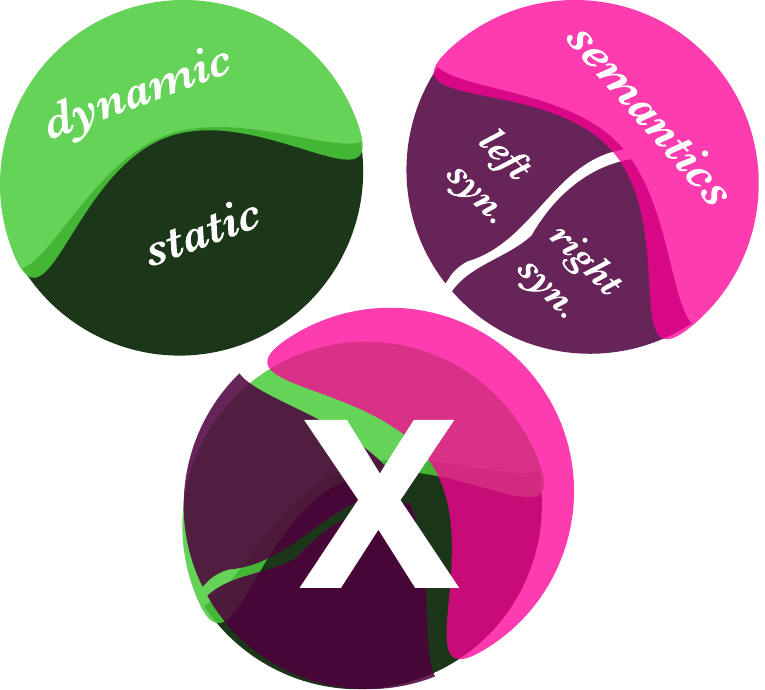}
  \end{center}
  \caption{A geometrical depiction of the topos $\XTop$ of parametricity structures: dark and light regions in the same color-range indicate complementary open and closed subtopoi corresponding to the static--dynamic and syntactic--semantic distinctions.}
  \label{fig:geometry}
\end{figure}

\subsection{Topo-logical metatheory of programming languages}

To prove a property of a logical system, it has been common practice
since the famous work of \citet{tarski-mckinsey:1946} and \citet{kripke:1965}
to interpret the logic into the preorder $\Opn{X}$ of opens of a carefully
chosen topological space $X$.  In this way, one may study a given axiom by
finding a space whose logic of opens either verifies or refutes it. One quickly
runs up against the limitations of this ``topo-logical'' approach, however: it
is not appropriate to interpret the terms $\Gamma\vdash a : A$ of a programming
language as morphisms in a preorder, because there exist non-equal $a,b : A$!

\paragraph{From opens to sheaves}

The problem identified above can be partly resolved by generalizing the concept
of an \emph{open} of a topological space to a \emph{sheaf} on a topological
space. While an open $U \in\Opn{X}$ can be thought of as a continuous mapping
from $X$ to the space of truth values, a sheaf $E : \Sh{X}$ can be thought of
as a continuous mapping from $X$ to the space of all sets. From this characterization, it is clear that sheaves
generalize opens, and one might hope this would make enough room for the
investigation of most type theoretic problems.

\paragraph{A category of points?}

Although the generalization to sheaves solves many problems for the study of
logic qua type theory, it is not enough. In programming languages one considers
semantics in functor categories $\ECat = \brk{\CCat,\SET}$ as in the work of
\citet{reynolds:1995,oles:1986}, but $\ECat$ is not likely to be of the form
$\Sh{X}$ for a topological space $X$ unless $\CCat$ is a preorder. The
geometric way to view this problem is as follows: if $\brk{\CCat,\SET}$ were
the category of sheaves on a topological space, the collection of points of
this space would have to form a category and not a preorder.

The preponderance of useful categories that behave \emph{as if they were} the
category of sheaves on a space led algebraic geometers under the leadership of
Grothendieck in the early 1970s to consider a new kind of generalized
space called a \textbf{topos} defined in terms of such categories, in which
the refinement relation between two points might be witnessed by non-trivial
evidence rather than being at most true~\citep{sga:4}.
The importance of this ``proof-relevance'' in geometry is
as follows: while there cannot be a topological space whose collection of
points is the category of local algebras for a given ring, there is a category
that behaves \emph{as if it were} the category of sheaves
on such a space, if it could exist.

\paragraph{Logoi and topoi}

What does it mean to behave like a category of sheaves on a space? The
behavioral properties of such a category, called a \emph{logos} by
\citet{anel-joyal:2021}, were concentrated by Giraud into a several simple
axioms.

\begin{definition}[Logos]
  A \emph{logos}, or category of sheaves, is a category closed under finite
  limits and small colimits, such that colimits commute with finite limits,
  sums are disjoint, and quotients are effective;\footnote{The condition that
  colimits commute with finite limits is analogous to the way that finite meets
  distribute over joins in $\Opn{X}$ for a topological space $X$.} for
  technical reasons one also requires that a logos be presentable by generators
  and relations.  A morphism between logoi is just a functor that preserves
  this structure, \ie finite limits and small colimits.
\end{definition}

Grothendieck's important idea was to take the (very large) category of logoi
and then define a new kind of space in terms of these, which he called the
\emph{topos}.

\begin{definition}[Topos]\label{def:topos}
  A topos $\XTop$ is defined by specifying a logos conventionally called
  $\Sh{\XTop}$, the category of ``sheaves on $\XTop$''; a continuous map of
  topoi $\Mor[f]{\XTop}{\YTop}$ is defined by specifying a morphism of logoi
  $\Mor[f^*]{\Sh{\YTop}}{\Sh{\XTop}}$ called the \emph{inverse image} of $f$,
  \ie a functor that is left exact (preserves finite limits) and cocontinuous
  (preserves colimits).  In this way, by definition, one has a contravariant
  equivalence $\Sh{-} : \Mor{\OpCat{\mathbf{Topos}}}{\mathbf{Logos}}$.
\end{definition}

\begin{remark}
  The left exactness and cocontinuity of morphisms of logoi generalizes the way that the
  inverse image of a continuous map between topological spaces preserves all
  joins and finite meets, as a morphism between frames of open sets.
\end{remark}

\begin{definition}[Direct image]
  For a morphism of topoi $\Mor[f]{\XTop}{\YTop}$, the cocontinuity of the
  inverse image $\Mor[f^*]{\Sh{\YTop}}{\Sh{\XTop}}$ implies that it is a left
  adjoint $f^*\dashv f_*$; the right adjoint $f_*$ is called the \emph{direct
  image}.
\end{definition}

The style of \cref{def:topos} is analogous to how a topological space is defined by
specifying what its open sets are! In the case of topoi sheaves play the role
that opens play in topological spaces.  A topological space $X$ gives rise to a topos
$\EnvTop{X}$, setting $\Sh{\EnvTop{X}}$ to be the classic category of sheaves on $X$;
but the language of topoi is more practical than the language of topological
spaces, because it contains more of the objects that we need in order to solve
type theoretic and logical problems.

\begin{example}
  The domain interpretation of programming languages can be seen to be an
  instance of this generalized ``topo-logical'' approach: while we are not
  aware of any topological space whose category of sheaves embeds the
  $\omega$-CPOs, it is possible to find a topos with this property, making the
  Scott semantics of programming languages a special case of sheaf
  semantics~\citep{fiore-rosolini:1997}.
\end{example}

\subsection{The language of topoi}

We give a crash course in the language of topoi insofar as it is
pertinent to the present paper. This section can be skipped and referred back
to by those comfortable with topoi; not all (or even most) of the material
presented here is necessary to understand our constructions, but we provide it
to assist the reader in developing topological intuitions for topoi which were
important for developing the present work. More details can be found in several
cited
resources~\citep{anel-joyal:2021,vickers:2007,johnstone:2002,stacks-project,wraith:1975}.

\begin{definition}[Subtopos]
  A subtopos $\XTop\subseteq\YTop$ is given by a logos $\Sh{\XTop}$ that is a
  subcategory of $\Sh{\YTop}$.
\end{definition}

\begin{definition}[Embedding]
  A morphism of topoi $\Mor[f]{\XTop}{\YTop}$ is called an \emph{embedding},
  written $\EmbMor[f]{\XTop}{\YTop}$, when the direct image $f_*$ functor is fully
  faithful.
\end{definition}

\begin{definition}[Equivalence]
  A morphism of topoi $\Mor[f]{\XTop}{\YTop}$ is called an \emph{equivalence}
  when the inverse image functor $f^*$ (equivalently, the direct image functor
  $f_*$) is an equivalence of categories.
\end{definition}

\begin{definition}[Opens of a topos]\label{def:open}
  An \emph{open} of a topos $\XTop$ is defined to be a subterminal object in
  $\Sh{\XTop}$, \ie a proof-irrelevant proposition in the internal type
  theory of $\XTop$. We will write $\Opn{\XTop}$ for the frame of opens of the
  topos $\XTop$.
  An open $U \in \Opn{\XTop}$ gives rise to an \emph{open subtopos}
  ${\XTop_U}\subseteq{\XTop}$: we define $\Sh{\XTop_U}$ to be the full subcategory
  of $\Sh{\XTop}$ spanned by objects $E$ such that the canonical map
  $\Mor{E}{E^U}$ is an isomorphism. Equivalently, $\Sh{\XTop_U}$ is the slice
  logos $\Sl{\Sh{\XTop}}{U}$.
\end{definition}

\begin{definition}[Open immersion]
  An embedding of topoi $\EmbMor[f]{\XTop}{\YTop}$ is called an \emph{open
  immersion}, written $\Mor|open embedding|[f]{\XTop}{\YTop}$ when it factors
  through an equivalence $\XTop\simeq\YTop_U$ and an open subtopos inclusion $\YTop_U\subseteq\YTop$
  for some open $U\in \Opn{\YTop}$ in the following sense:
  \[
    \begin{tikzpicture}[diagram]
      \node (0) {$\XTop$};
      \node (2) [right = of 0] {$\YTop$};
      \node (1) [between = 0 and 2,yshift = -1.5cm] {$\YTop_U$};
      \path[->] (0) edge node [sloped,below] {$\simeq$} (1);
      \path[open embedding] (0) edge node [above] {$f$} (2);
      \path[open embedding] (1) edge (2);
    \end{tikzpicture}
  \]
\end{definition}

\begin{definition}[Closed complement]\label{def:closed-complement}
  Let $U\in \Opn{\XTop}$ be an open of a topos $\XTop$; the \emph{closed
  complement} of the open subtopos $\XTop_U$ can be defined by means of the
  full subcategory $\Sh{\XTop}\Sup{\lor U} \subseteq \Sh{\XTop}$ spanned by objects $E$ such that the
  canonical map $\Mor{E}{E\sqcup\Sub{E\times U}U}$ is an isomorphism, where
  $E\sqcup\Sub{E\times U} U$ is the following pushout:
  \[
    \DiagramSquare{
      nw = E\times U,
      sw = E,
      ne = U,
      se = E\sqcup\Sub{E\times U}U,
      se/style = pushout,
    }
  \]

  Then the closed complement $\XTop\Sub{\setminus U}\subseteq \XTop$ is defined by the
  identification $\Sh{\XTop\Sub{\setminus{U}}} = \Sh{\XTop}\Sup{\lor U}$.
\end{definition}

It is not hard to show that a sheaf on the closed complement
$\XTop\Sub{\setminus{U}}$ is the same as a sheaf $E$ on $\XTop$ that is
\emph{$U$-connected} in the sense that $E\times U \cong U$, or equivalently, $E^U
\cong \ObjTerm{\Sh{\XTop}}$.

\begin{definition}[Closed immersion]\label{def:closed-immersion}
  Likewise an embedding of topoi is called a \emph{closed} immersion, written
  $\Mor|closed embedding|[f]{\XTop}{\YTop}$ when it factors through an
  equivalence $\XTop\simeq\YTop\Sub{\setminus{U}}$ and a closed subtopos inclusion $\YTop\Sub{\setminus U} \subseteq
  \YTop$ for some open $U\in\Opn{\YTop}$.
\end{definition}

The open and closed subtopoi corresponding to $U\in\Opn{\XTop}$ are
complementary in the sense of classical topology, but this does not mean there
is no substance lying between them. Just as in classical topology, between an open
subspace and its closed complement lies a ``boundary'' $\partial_\XTop{U} =
\overline{\XTop_U}\cap \XTop\Sub{\setminus U}\subseteq\XTop$ where
$\overline{U}$ is the closure of $U$.

\begin{definition}[Closure of an open subtopos]
  If $U\in\Opn{\XTop}$ is an open of a topos $\XTop$, the \emph{closure} of $U$
  is the smallest closed subtopos $\overline{\XTop_U}\subseteq \XTop$ that
  contains $\XTop_U$. Writing $\Mor|open embedding|[\OpEmb]{\XTop_U}{\XTop}$
  for the open immersion corresponding to $U$, we have a lex idempotent monad
  $\Mor[\OpEmb_*\OpEmb^*]{\Opn{\XTop}}{\Opn{\XTop}}$ on the frame of opens
  given by the adjunction $\OpEmb^*\dashv\OpEmb_*$. Considering the
  characterization of $\Sh{\XTop_U}$ as the slice $\Sl{\Sh{\XTop}}{U}$, we see
  that $\OpEmb_*\OpEmb^*V$ is the Heyting implication $\prn{U\Rightarrow V}$
  for any $V\in \Opn{\XTop}$.
  The \emph{closure} of $U$ can then be computed to be the closed complement of
  the open $\OpEmb_*\OpEmb^*\bot = \lnot{U}$, \ie the closed subtopos
  $\XTop\Sub{\setminus\lnot{U}}\subseteq\XTop$. Explicitly, a sheaf on
  $\XTop\Sub{\setminus\lnot{U}}$ is a sheaf on $\XTop$ that is
  $\lnot{U}$-connected, \ie becomes a singleton when restricted to
  $\XTop\Sub{\lnot{U}}$.
\end{definition}

\begin{definition}[Fringe of an open subtopos]\label{def:fringe}
  Let $U\in\Opn{\XTop}$ be an open of a topos $\XTop$; the \emph{fringe}
  $\partial_\XTop{U}\subseteq\XTop$ of the open subtopos $\XTop_U$ is defined
  to be the intersection of the closure of $\XTop_U$ with the closed complement
  $\XTop\Sub{\setminus{U}}$.
  \[
    \DiagramSquare{
      nw/style = pullback,
      nw = \partial_\XTop{U},
      ne = \overline{\XTop_U}\mathrlap{{}=\XTop\Sub{\setminus{\lnot{U}}}},
      se = \XTop,
      sw = \XTop\Sub{\setminus{U}},
      south/style = closed embedding,
      east/style = closed embedding,
      west/style = {closed embedding,exists},
      north/style = {closed embedding,exists},
    }
  \]

  By a further computation, we may observe that the fringe
  $\partial_\XTop{U}$ is the closed subtopos corresponding to the open
  $U\lor\lnot{U}\in\Opn{\XTop}$, \ie we have $\partial_\XTop{U} =
  \XTop\Sub{\setminus U\lor\lnot{U}}$. This is not trivial unless $U$ is
  simultaneously closed and open!
\end{definition}

\begin{remark}
  The above shows the geometric sense in which the open and closed subtopoi are
  complementary; although we always have $\OpEmb_*\OpEmb^*\ClEmb_*\ClEmb^*E \cong
  \ObjTerm{\Sh{\XTop}}$, we do not have
  $\ClEmb_*\ClEmb^*\OpEmb_*\OpEmb^*E\cong\ObjTerm{\Sh{\XTop}}$ except when $U$ is
  clopen.
\end{remark}

\begin{definition}[The fringe functor]\label{def:fringe-functor}
  Given an open $U\in\Opn{\XTop}$, define what is called the \emph{fringe
  functor} $\Mor[F_{U}]{\Sh{\XTop_U}}{\Sh{\XTop\Sub{\setminus{U}}}}$ to be the
  following composite:
  \[
    \begin{tikzpicture}[diagram]
      \node (0) {$\Sh{\XTop_U}$};
      \node (1) [right = of 0] {$\Sh{\XTop}$};
      \node (2) [right = of 1] {$\Sh{\XTop\Sub{\setminus{U}}}$};
      \path[->] (0) edge node [above] {$\OpEmb_*$} (1);
      \path[->] (1) edge node [above] {$\ClEmb^*$} (2);
      \path[->,exists,bend right=30] (0) edge node [below] {$F_U$} (2);
    \end{tikzpicture}
  \]
\end{definition}

The relationship between the fringe functor corresponding to $U\in\Opn{\XTop}$
(\cref{def:fringe-functor}) and the fringe of the open subtopos
$\XTop_U\subseteq\XTop$ (\cref{def:fringe}) is expressed in the
\cref{lem:fringe-vs-fringe,thm:recollement} below.

\begin{lemma}[\citet{wraith:1975}]\label{lem:fringe-vs-fringe}
  If $E$ is a sheaf on $\XTop_U$, the sheaf $F_U{E}$ on
  $\XTop\Sub{\setminus{U}}$ is trivial away from the fringe
  $\partial_\XTop{U}\subseteq\XTop\Sub{\setminus{U}}$, \ie in $F_U{E}$
  restricts to the terminal sheaf in the open complement of
  $\partial_\XTop{U}\subseteq\XTop\Sub{\setminus{U}}$.
\end{lemma}

\begin{proof}
  As a closed subtopos of $\KTop := \XTop\Sub{\setminus{U}}$, the fringe
  $\partial_\XTop{U}$ is the complement of the open
  $\ClEmb^*\prn{U\lor\lnot{U}}$, which is equal to $\ClEmb^*\lnot{U}$ because
  inverse image is cocontinuous and $\ClEmb^*U = \bot$.  Therefore we may
  reconstruct $\partial_\XTop{U}\subseteq\KTop$ as
  $\KTop\Sub{\setminus{\ClEmb^*\lnot{U}}}$ and our goal is to show
  that $F_UE\in \Sh{\KTop\Sub{\setminus\ClEmb^*\lnot{U}}}\subseteq\Sh{\KTop}$, which is the same as to show
  $F_U{E}\times \ClEmb^*\lnot{U} \cong \ClEmb^*\lnot{U}$.
  \begin{align*}
    F_U{E}\times \ClEmb^*\lnot{U} &\cong
    \ClEmb^*\OpEmb_*{E}\times \ClEmb^*\lnot{U}
    &&\text{\cref{def:fringe-functor}}
    \\
    &\cong
    \ClEmb^*\prn{\OpEmb_*{E}\times \lnot{U}}
    &&\text{$\ClEmb^*$ lex}
    \\
    &\cong
    \ClEmb^*\prn{\OpEmb_*\OpEmb^*\OpEmb_*{E}\times \lnot{U}}
    &&\text{$\OpEmb_*\OpEmb^*$ idempotent}
    \\
    &\cong
    \ClEmb^*\prn{\DelimMin{1}\prn{\OpEmb_*{E}}^U\times \lnot{U}}
    &&\text{$\OpEmb_*\OpEmb^* \cong \prn{-}^U$}
    \\
    &\cong
    \ClEmb^*\lnot{U}
    &&
    \qedhere
  \end{align*}
\end{proof}

\begin{theorem}[Artin gluing / Recollement~\citep{sga:4}]\label{thm:recollement}
  A topos $\XTop$ can be reconstructed up to equivalence from the data of a
  partition $U\in\Opn{\XTop}$ into open and closed subtopoi:\footnote{The construction
  takes place in the category of categories and functors, rather than the
  category of logoi and morphisms of logoi; this is because the fringe funtor $F_U$ need not be cocontinuous.}
  \[
    \DiagramSquare{
      ne = \ArrCat{\Sh{\XTop\Sub{\setminus{U}}}},
      se = {\Sh{\XTop\Sub{\setminus{U}}}},
      east = \Cod,
      sw = \Sh{\XTop_U},
      south = F_U,
      nw/style = pullback,
      west = \OpEmb^*,
      nw = \mathllap{\Sh{\XTop}\simeq{}}\brc{\Sh{\XTop\Sub{\setminus{U}}}}\downarrow F_U,
      width = 3cm,
    }
  \]

  Conversely if $\Mor[F]{\Sh{\UTop}}{\Sh{\KTop}}$ is a left exact and
  accessible functor between logoi, then there exists a topos $\XTop$ together
  with an open $U\in\Opn{\XTop}$ such that $\UTop\simeq \XTop_U$ and
  $\KTop\simeq\XTop\Sub{\setminus{U}}$ configured like so:
  \[
    \DiagramSquare{
      ne = \ArrCat{\Sh{\KTop}},
      se = {\Sh{\KTop}},
      east = \Cod,
      sw = \Sh{\UTop},
      south = F,
      nw/style = pullback,
      west = \OpEmb^*,
      nw = \Sh{\XTop},
    }
  \]

  Above, $\XTop$ is called the \emph{Artin gluing} of $F$; the open
  $U\in\Opn{\XTop}$ that reconstructs $\prn{\UTop,\KTop}$ as
  $\prn{\XTop_U,\XTop\Sub{\setminus{U}}}$ can be defined to be the subterminal
  sheaf $\prn{\ObjTerm{\Sh{\UTop}},
  \Mor{\ObjInit{\Sh{\KTop}}}{F\prn{\ObjTerm{\Sh{\UTop}}}}}$.
\end{theorem}

\begin{remark}[Geometric gluing]
  In certain cases, including those investigated in this paper, the fringe
  functor $F_U$ turns out to be either the direct image or inverse image part of
  a morphism of topoi; in those cases, a construction of the Artin gluing taking
  place in the category of \emph{topoi} rather than the category of categories is
  available. When $F$ is the global sections functor, the Artin gluing is called
  the \emph{scone} or \emph{Sierpi\'nski cone}; when $F = f_*$ (resp.\ $F=f^*$)
  for a morphism of topoi $\Mor[f]{\YTop}{\ZTop}$, the Artin gluing is referred
  to by \citet{johnstone:topos:1977} as the open (resp.\ closed) mapping cylinder of
  $f$, depicted below:
  \[
    \DiagramSquare{
      nw = \YTop,
      ne = \ZTop,
      sw = \YTop\times\SP,
      se = \TopIdent{M}_f^\PtO,
      north = f,
      se/style = pushout,
      west/style = closed embedding,
      east/style = closed embedding,
      west = \prn{\ArrId{},\PtC},
      east = \ClEmb,
    }
    \qquad
    \DiagramSquare{
      nw = \YTop,
      ne = \ZTop,
      sw = \YTop\times\SP,
      se = \TopIdent{M}_f^\PtC,
      north = f,
      se/style = pushout,
      west/style = open embedding,
      east/style = open embedding,
      west = \prn{\ArrId{},\PtO},
      east = \OpEmb,
    }
  \]

  Above we have $\Sh{\TopIdent{M}_f^\PtO} \simeq \brc{\Sh{\ZTop}}\downarrow
  f_*$ and $\Sh{\TopIdent{M}_f^\PtC}\simeq\brc{\Sh{\YTop}}\downarrow f^*$.
  Indeed, our \cref{con:xtop} of the topos of parametricity structures in
  \cref{sec:parametricity-topos} is an example of the open mapping cylinder
  depicted above on the left.
\end{remark}

\subsubsection{Classifying topoi and geometric figures}

The same topos $\XTop$ can be profitably understood in two \emph{different}
ways: what happens when you map into it, and what happens when you map out of
it. These two perspectives correspond respectively to viewing a topos as
\emph{classifying space} of some kind of data, \emph{vs.}\ as a \emph{geometrical
figure}; these correspond to algebraic and geometrical perspectives on topoi
respectively. In other words, a morphism of topoi $\Mor{\XTop}{\YTop}$ can be
thought of as constructing a ``point'' of $\YTop$ in the language of
$\Sh{\XTop}$, but it can also be thought of as drawing an $\XTop$-shaped figure
in $\YTop$.
We use both perspectives in this paper in a critical way; in particular, the
Sierpi\'nski topos $\SP$ appears in our construction both as a geometrical
figure and as a classifier (see \cref{rem:sierp-duality}).

\begin{example}[The punctual topos]
  The logos of sets $\SET$ is the category of sheaves on the one-point space.
  Therefore we define the \emph{punctual} topos $\PtTop : \mathbf{Topos}$ to
  by the identification $\Sh{\PtTop} = \SET$.
  From the geometrical point of view, a morphism $\Mor{\PtTop}{\XTop}$
  corresponds to constructing a \emph{point} of $\XTop$. From the algebraic
  point of view, one thinks of a morphism $\Mor{\XTop}{\PtTop}$ as constructing
  no data whatsoever in the language of $\Sh{\XTop}$. For this reason, there
  is always a unique such morphism and hence $\PtTop$ is the terminal topos.
\end{example}

\begin{example}[Presheaves and finite limit theories]\label{ex:psh}
  Let $\CCat$ be a small category; then $\Psh{\CCat}$ is the category of
  \emph{presheaves} on $\CCat$, \ie functors $\Mor{\OpCat{\CCat}}{\SET}$. We
  write $\PrTop{\CCat}$ for the topos whose sheaves are the presheaves on
  $\CCat$, \ie $\Sh{\PrTop{\CCat}} = \Psh{\CCat}$.
  Suppose that $\CCat$ has finite limits, \ie $\CCat$ is the classifying
  category for a finite limit theory $\mathbb{T}$; then Diaconescu's
  theorem~\citep{diaconescu:1975} states that a morphism
  $\Mor{\YTop}{\PrTop{\CCat}}$ corresponds to a left exact functor
  $\Mor{\CCat}{\Sh{\YTop}}$, \ie a model of $\mathbb{T}$ in $\Sh{\YTop}$.

  Hence the algebraic perspective says that $\PrTop{\CCat}$ is the classifier
  of $\mathbb{T}$-models. On the other hand, $\PrTop{\CCat}$ is a geometric
  figure that captures the configuration of \emph{all} $\mathbb{T}$-models and
  their homomorphism. It is helpful to consider the case where $\mathbb{T}$ is
  the theory of groups: then, for example, a morphism of topoi
  $\Mor{\PrTop{\CCat}}{\SP}$ corresponds to a diagram of truth values that is
  labeled by the collection of \emph{all} groups and group homomorphisms.
\end{example}

\begin{example}[Sierpi\'nski topos]\label{ex:sierp}
  The logos of \emph{families} of sets $\ArrCat{\SET}$ is also the category of
  sheaves on the classic Sierpi\'nski space $S =
  \brc{\brc{},\brc{\PtO},\brc{\PtC,\PtO}}$. Hence we define the Sierpi\'nski
  \emph{topos} $\SP$ by the identification $\Sh{\SP} = \ArrCat{\SET}$. As a
  geometrical figure, the Sierpi\'nski topos is a \emph{directed interval} in
  that it has two points $\Mor[\PtC,\PtO]{\PtTop}{\SP}$ and a morphism
  $\Mor{\PtC}{\PtO}$.  A morphism $\Mor{\SP}{\XTop}$ corresponds to a pair of
  points $\Mor[x,y]{\PtTop}{\XTop}$ together with a morphism of points
  $\Mor{x}{y}$. The open point determines a distinguished open that we
  might write $\brc{\PtO}\in\Opn{\SP}$, which is defined by the subterminal
  family of sets $\Mor{\ObjInit{\SET}}{\ObjTerm{\SET}} : \ArrCat{\SET}$.

  From the algebraic point of view, $\SP$ classifies
  \emph{opens} or \emph{propositions} in that every open subtopos
  $\XTop_U\subseteq \XTop$ arises in an essentially unique way by pullback
  along the open point
  $\OpnEmbMor[\PtO]{\PtTop}{\SP}$:
  \[
    \DiagramSquare{
      ne = \PtTop,
      se = \SP,
      sw = \XTop,
      nw = \XTop_U,
      south = \floors{U},
      east = \PtO,
      south/style = exists,
      west/style = open embedding,
      east/style = open embedding,
      nw/style = pullback,
    }
  \]

  The characteristic map $\floors{U}$ has a universal property in the inverse
  image direction, namely it is the unique map such that
  $\floors{U}^*\brc{\PtO} = U \in \Opn{\XTop}$. A geometric/pointwise intuition
  is also helpful: the characteristic map $\Mor[\floors{U}]{\XTop}{\SP}$ sends
  a point $x\in \XTop$ to the open point $\PtO\in \SP$ if $x\in \XTop_U$, and
  sends it to the closed point $\PtC$ if $x\not\in\XTop_U$.
\end{example}

\begin{computation}
  Given an open $U\in\Opn{\XTop}$, how is the corresponding characteristic map
  $\Mor[\floors{U}]{\XTop}{\SP}$ actually constructed? Dualizing into the
  language of logoi, we must construct a lex and cocontinous functor
  $\Mor[\floors{U}^*]{\Sh{\SP}}{\Sh{\XTop}}$. Note that $\Sh{\SP}$ is the
  category of presheaves on the interval category $\brk{1} =
  \brc{\Mor|>->|{0}{1}}$ and that moreover, $\brk{1}$ has finite limits. Hence
  as discussed in \cref{ex:psh}, Diaonescu's theorem states that a lex and
  cocontinuous functor $\Mor{\Sh{\SP}}{\Sh{\XTop}}$ is the same thing as a lex
  functor $\Mor{\brk{1}}{\Sh{\XTop}}$, which can be seen to be the same thing
  as a subterminal object in $\Sh{\XTop}$, \ie an open of $\XTop$.

  A slightly more elementary way to understand what is happening here is to
  observe that $\Sh{\SP}$ is the free cocompletion of $\brk{1}$, so a
  cocontinuous morphism out of $\Sh{\SP}$ has freedom only in the base case: it
  must take formal colimits of $\brk{1}$ to actual colimits of $\Sh{\XTop}$.
\end{computation}

It is helpful to investigate the morphisms $\Mor[\PtO,\PtC]{\PtTop}{\SP}$ in
terms of the geometry--algebra duality, which we depict in \cref{tab:sierp-points}.

\begin{table}
  \begin{center}
    \begin{tabular}{lll}
      \toprule
       & \textbf{Geometric perspective} & \textbf{Algebraic perspective}\\
      \midrule
      $\Mor|closed embedding|[\PtC]{\PtTop}{\SP}$ & the closed point of $\SP$ & the proposition $\bot$ in $\SET$\\
      $\Mor|open embedding|[\PtO]{\PtTop}{\SP}$ & the open point of $\SP$ & the proposition $\top$ in $\SET$\\
      $\Mor{\PtC}{\PtO}$ & the directed interval & the implication $\bot\leq \top$\\
      \bottomrule
    \end{tabular}
  \end{center}

  \caption{Investigating the points of the Sierpi\'nski topos from the
  geometrical (topos as figure) and algebraic perspectives (topos as
  classifier).}
  \label{tab:sierp-points}
\end{table}

\subsection{Phase separation and the Sierpi\'nski topos}

We intend to use the Sierpi\'nski topos $\SP$ to capture the notion of phase
separation: in essence, a sheaf on $\SP$ will be a kind of ``phase separated
set''. To substantiate this intuition, we must consider an explicit
construction of $\SP$ that allows us to characterize its sheaves in terms of
something familiar.

\begin{computation}\label{comp:sp-top}

  If a sheaf on $\SP$ is just a family of sets, then we may profitably view the
  downstairs part of such a family as its ``static component'', the upstairs part
  as its ``dynamic component''; the projection expresses the dependency of
  dynamic on static.  The inverse image of the open point $\OpnEmbMor[\PtO]{\PtTop}{\SP}$ is
  the codomain functor $\Mor[\Cod]{\ArrCat{\SET}}{\SET}$, and the inverse image
  of the closed point $\ClEmbMor[\PtC]{\PtTop}{\SP}$ is the domain functor
  $\Mor[\Dom]{\ArrCat{\SET}}{\SET}$.
\end{computation}

Of course, we might equally well replace the (static, dynamic) intuition with
(syntactic, semantic), reflecting the fact that splitting a logical relation
into syntactic and semantic parts is \emph{itself} a kind of phase distinction
in the language of logical relations.  For this reason logical relations for a
calculus that admits a phase distinction can be thought of as an
iteration of logical relations: the underlying calculus \ModTT{} is
already a language of (proof-relevant) synthetic logical
relations over the sublanguage of purely static kinds and constructors.

\subsubsection{Phase separated global sections}\label{sec:ph-gsec}

Let $\ThCat$ be the syntactic category of \ModTT{}; we may
manipulate $\ThCat$ in the language of topoi by enlarging it to
$\PrTop{\ThCat}$, the topos of presheaves on $\ThCat$ (see \cref{ex:psh}).
$\PrTop{\ThCat}$ can be thought of as a topos of generalized syntax.

We consider the characteristic map $\Mor[\PhGS]{\PrTop{\ThCat}}{\SP}$
of the open $\StOpn \in \Opn{\PrTop{\ThCat}}$, so we have
$\PhGS^*\brc{\PtO} = \StOpn$.
We will see in \cref{comp:ph-gs-dir-img} that the direct image
$\Mor[\PhGS_*]{\Psh{\ThCat}}{\Sh{\SP}}$ can be viewed as a ``phase separated'' version of the
global sections functor, sending each object to the weakening map from its
closed elements to their static parts.

\begin{computation}\label{comp:ph-gs-dir-img}
  To verify the intuition above, we proceed to compute the action of the direct
  image $\PhGS_*$ on a presheaf $X:\Psh{\ThCat}$. First, we recognize that the
  direct image $\PhGS_*X$ should be a family of sets (\ie a
  $\OpCat{\Simplex{1}}$-shaped diagram of sets) by definition; we probe
  this family of sets at the map $\Mor|>->|{0}{1} : \Simplex{1}$ using the
  Yoneda lemma, adjointness, and the fact that $\PhGS$ is the characteristic
  map of the open $\StOpn$:
  \begin{align*}
    \PhGS_*X\brc{\Mor|>->|{0}{1}} &\cong
    \Hom[\Sh{\SP}]{\Yo[\Simplex{1}]\brc{\Mor|>->|{0}{1}}}{\PhGS_*X}
    &&\text{by Yoneda lemma}
    \\
    &\cong
    \Hom[\Sh{\SP}]{\brc{\Mor|>->|{\brc{\PtO}}{\ObjTerm{\Sh{\SP}}}}}{\PhGS_*X}
    &&\text{by \cref{ex:sierp}}
    \\
    &\cong
    \Hom[\Psh{\ThCat}]{\PhGS^*\brc{\Mor|>->|{\brc{\PtO}}{\ObjTerm{\Sh{\SP}}}}}{X}
    &&\text{by $\PhGS^*\dashv\PhGS_*$}
    \\
    &\cong
    \Hom[\Psh{\ThCat}]{\brc{\Mor|>->|{\StOpn}{\ObjTerm{\Psh{\ThCat}}}}}{X}
    &&\text{by def.\ of $\PhGS$}
  \end{align*}

  Hence $\PhGS_*X$ is the diagram $\Mor{X\prn{\ObjTerm{\ThCat}}}{X\prn{\StOpn}}$ of sets
  that projects from a global element (closed term) of $X$ its static part.
\end{computation}

\subsection{Topos of parametricity structures}\label{sec:parametricity-topos}

We will construct a topos whose sheaves will model the parametricity
structures of \PSTT, as proof-relevant relations between two potentially different
syntactic objects.
Let $E$ be a finite cardinal and $\YTop$ a topos. The copower $E\cdot\YTop = \Coprod{e\in
E}{\YTop}$ is a topos, whose corresponding logos may be computed as follows:
$
  \Sh{E\cdot\YTop} = \Sh{\Coprod{e\in E}{\YTop}}= \Prod{e\in E}{\Sh{\YTop}} =
\Sh{\YTop}^E
$.

The codiagonal morphism of topoi $\Mor[\nabla]{E\cdot\YTop}{\YTop}$
corresponds under inverse image to the diagonal morphism of logoi
$\Mor[\nabla^*]{\Sh{\YTop}}{\Sh{\YTop}^E}$; indeed, the diagonal map is lex as
it is right adjoint to the colimit functor
$\Mor[\Colim{E}]{\Sh{\YTop}^E}{\Sh{\YTop}}$, and it is cocontinuous because it
is left adjoint to the limit functor, \ie the direct image
${\color{gray}\nabla^*}\dashv\nabla_*$.
Because we are considering binary parametricity, we will set $E := 2$
and define a topos whose sheaves correspond to parametricity structures by
gluing. We may consider the following morphism
$\Mor[\FR]{2\cdot\PrTop{\ThCat}}{\SP}$ of topoi:
\begin{equation}\label[diagram]{diag:fringe-functor}
  \begin{tikzpicture}[diagram]
    \node (0) {$2\cdot\PrTop{\ThCat}$};
    \node (1) [right = of 0] {$\PrTop{\ThCat}$};
    \node (2) [right = of 1] {$\SP$};
    \path[->] (0) edge node [above] {$\nabla$} (1);
    \path[->] (1) edge node [above] {$\PhGS$} (2);
    \path[->,exists,bend right = 30] (0) edge node [below] {$\FR$} (2);
  \end{tikzpicture}
\end{equation}

\begin{computation}\label{cmp:fringe-dirimg}
  The direct image $\Mor[\FR_*]{\Psh{\ThCat}^2}{\Sh{\SP}}$ takes a pair
  $\prn{X_L,X_R} : \Psh{\ThCat}^2$ of (generalized) syntactic objects to
  $\PhGS_*X_L\times \PhGS_*X_R$, the product of their phase separated global
  sections.
\end{computation}

\begin{proof}
  To see that this is the case, we first dualize \cref{diag:fringe-functor}
  into the language of logoi.
  \begin{equation}\label[diagram]{diag:fringe-functor:invimg}
    \begin{tikzpicture}[diagram]
      \node (0) {$\Psh{\ThCat}^2$};
      \node (1) [left = of 0] {$\Psh{\ThCat}$};
      \node (2) [left = of 1] {$\Sh{\SP}$};
      \path[<-] (0) edge node [above] {$\nabla^*$} (1);
      \path[<-] (1) edge node [above] {$\PhGS^*$} (2);
      \path[<-,exists,bend left = 30] (0) edge node [below] {$\FR^*$} (2);
    \end{tikzpicture}
  \end{equation}

  In \cref{diag:fringe-functor:invimg} above, the inverse image $\nabla^*$ is
  the diagonal functor and hence its right adjoint $\nabla_*$ is the product functor.
  Hence we may compute the direct image part of $\FR$ as follows:
  \begin{equation}\label[diagram]{diag:fringe-functor:dirimg}
    \begin{tikzpicture}[diagram]
      \node (0) {$\Psh{\ThCat}^2$};
      \node (1) [right = of 0] {$\Psh{\ThCat}$};
      \node (2) [right = of 1] {$\Sh{\SP}$};
      \path[->] (0) edge node [above] {$\prn{\times}$} (1);
      \path[->] (1) edge node [above] {$\PhGS_*$} (2);
      \path[->,exists,bend right = 30] (0) edge node [below] {$\FR_*$} (2);
    \end{tikzpicture}
  \end{equation}

  Because $\PhGS_*$ is continuous, we may commute it past the product functor:
  \begin{equation}
    \begin{tikzpicture}[diagram,baseline=(sw.base)]
      \SpliceDiagramSquare{
        width = 2.25cm,
        nw = \Psh{\ThCat}^2,
        sw = \Psh{\ThCat},
        ne = \Sh{\SP}^2,
        se = \Sh{\SP},
        north = \PhGS_*^2,
        east = \prn{\times},
        west = \prn{\times},
        south = \PhGS_*,
      }
      \path[->] (nw) edge node [desc] {$\FR_*$} (se);
    \end{tikzpicture}
  \end{equation}

  Hence $\FR_*$ takes a pair $\prn{X_L,X_R} : \Psh{\ThCat}^2$ to
  $\PhGS_*X_L\times \PhGS_*X_R$.
\end{proof}

\begin{construction}[Topos of parametricity structures]\label{con:xtop}
  We then obtain a topos $\XTop$ whose sheaves correspond to parametricity
  structures by gluing, specifically via a phase separated version of the
  Sierpi\'nski cone construction: we first form the Sierpi\'nski cylinder
  $\prn{2\cdot\PrTop{\ThCat}}\times \SP$ and then \emph{pinch} the end
  corresponding to the closed point $\PtC\in \SP$ along $\FR$ as follows:
  \[
    \begin{tikzpicture}[diagram]
      \SpliceDiagramSquare<sq/>{
        se/style = pushout,
        east/style = {closed embedding, exists, color = RedDevil},
        west/style = closed embedding,
        south/style = exists,
        east = \ClEmb,
        nw = 2\cdot\PrTop{\ThCat},
        sw = \prn{2\cdot \PrTop{\ThCat}}\times \SP,
        se = \XTop,
        ne = \SP,
        west = \prn{\ArrId{}, \PtC},
        north = \FR,
        width = 2.25cm,
        height = 2.25cm,
      }
      \node (sw) [below = 2.25cm of sq/sw] {$2\cdot\PrTop{\ThCat}$};
      \path[open embedding] (sw) edge node [left] {$\prn{\ArrId{},\PtO}$} (sq/sw);
      \path[open embedding,exists,color = RegalBlue] (sw) edge node [sloped,below] {$\OpEmb$} (sq/se);
    \end{tikzpicture}
  \]

  The induced embedding $\Mor|open
  embedding|[\OpEmb]{2\cdot\PrTop{\ThCat}}{\XTop}$ can be seen to be an open
  immersion; moreover, its image is the open complement of the image of the
  closed immersion. Therefore $\XTop$ is a topos governing parametricity
  structures, and restricting along the open immersion projects the (doubled)
  syntactic part of a parametricity structure, whereas restricting along the
  closed immersion projects the (phase separated) semantic part of a
  parametricity structure.
\end{construction}

\begin{remark}\label{rem:sierp-duality}
  The Sierpi\'nski topos $\SP$ plays two distinct roles in \cref{con:xtop}:
  first, we use $\SP$ to form a cylinder on $2\cdot\PrTop{\ThCat}$ (which is
  always done in gluing), and secondly $\SP$ is the codomain of the functor we
  are gluing along. In the first case, $\SP$ is acting as a directed interval figure
  whereas in the second case, $\SP$ is acting as the classifier of opens and
  $\FR$ is the characteristic map for the open of $2\cdot\PrTop{\ThCat}$
  that restricts on each side to the static open $\StOpn$.
  This second use corresponds to the fact that we
  are constructing \emph{phase separated} parametricity structures rather than
  ordinary parametricity structures, in which case we would be gluing into the
  punctual topos $\PtTop$.
\end{remark}

\begin{computation}
  We may compute an explicit description of parametricity structures, \ie
  sheaves on $\XTop$. A parametricity structure $X : \Sh{\XTop}$ is given by
  the following data:
  \begin{enumerate}
    \item A pair of generalized syntactic objects $X^\PtO_L,X^\PtO_R : \Psh{\ThCat}$.
    \item A family of phase separated sets $\Mor{X^\PtC}{\PhGS_*{X^\PtO_L}\times\PhGS_*{X^\PtO_R}}:\Sh{\SP}$, \ie a proof-relevant relation between the (phase separated) closed terms of $X^\PtO_L$ and $X^\PtO_R$.
  \end{enumerate}
\end{computation}

\begin{proof}
  We recall the pushout of topoi that defines $\XTop$ from \cref{con:xtop}.
  \begin{equation}\label[diagram]{diag:dualizing:0}
    \begin{tikzpicture}[diagram,baseline=(sw.base)]
      \SpliceDiagramSquare{
        se/style = pushout,
        east/style = {closed embedding},
        west/style = closed embedding,
        east = \ClEmb,
        nw = 2\cdot\PrTop{\ThCat},
        sw = \prn{2\cdot \PrTop{\ThCat}}\times \SP,
        se = \XTop,
        ne = \SP,
        west = \prn{\ArrId{}, \PtC},
        north = \FR,
        width = 2.25cm,
        height = 2.25cm,
      }
    \end{tikzpicture}
  \end{equation}

  We translate \cref{diag:dualizing:0} into the language of logoi, at first
  only dualizing:
  \begin{equation}\label[diagram]{diag:dualizing:1}
    \begin{tikzpicture}[diagram,baseline=(sw.base)]
      \SpliceDiagramSquare{
        nw/style = pullback,
        nw = \Sh{\XTop},
        se = \Sh{2\cdot\PrTop{\ThCat}},
        ne = \Sh{\prn{2\cdot \PrTop{\ThCat}}\times \SP},
        sw = \Sh{\SP},
        west = \ClEmb^*,
        east = \prn{\ArrId{}, \PtC}^*,
        south = \FR^*,
        width = 3.25cm,
        height = 2.25cm,
      }
    \end{tikzpicture}
  \end{equation}

  First we observe that the a sheaf on $\YTop\times\SP$ is the same as a family
  of sheaves on $\YTop$, and that the closed point is the coordinate for the
  \emph{domain} of such a family. Therefore we may rewrite the right-hand map
  of \cref{diag:dualizing:1} as follows:
  \begin{equation}\label[diagram]{diag:dualizing:2}
    \begin{tikzpicture}[diagram,baseline=(sw.base)]
      \SpliceDiagramSquare{
        nw/style = pullback,
        nw = \Sh{\XTop},
        se = \Sh{2\cdot\PrTop{\ThCat}},
        ne = \ArrCat{\Sh{2\cdot \PrTop{\ThCat}}},
        sw = \Sh{\SP},
        west = \ClEmb^*,
        east = \Dom,
        south = \FR^*,
        width = 3cm,
        height = 2.25cm,
      }
    \end{tikzpicture}
  \end{equation}

  By \cref{diag:dualizing:2}, we see that a sheaf on $\XTop$ carries the data
  of a sheaf $X^\PtC$ on $\SP$, a sheaf $X^\PtO = \prn{X^\PtO_L,X^\PtO_R}$ on
  $2\cdot\PrTop{\ThCat}$, and a morphism $\Mor{\FR^*X^\PtC}{X^\PtO} :
  \Sh{2\cdot\PrTop{\ThCat}}$. By adjoint transpose, this is the same
  a morphism $\Mor{X^\PtC}{\FR_*X^\PtO}$ and by \cref{cmp:fringe-dirimg}, this
  is a morphism $\Mor{X^\PtC}{\PhGS_*X^\PtO_L\times \PhGS_*X^\PtO_R}$.
\end{proof}

The open immersion $\OpnEmbMor[\OpEmb]{2\cdot\PrTop{\ThCat}}{\XTop}$ corresponds (by definition)
to an open $\SynOpn \in \Opn{\XTop}$, \ie the subterminal parametricity structure
$
  \SynOpn = \prn{\Mor{\ObjInit{\Sh{\SP}}}{\FR_*\prn{\ObjTerm{\Psh{\ThCat}^2}}}}
$.
Let $\YTop$ be a topos and $E$ a finite cardinal; the injections
$\OpnEmbMor[\mathsf{inj}_e]{\YTop}{E\cdot\YTop}$ into the coproduct are in fact
open immersions~\citep[Lemma~B.3.4.1]{johnstone:2002}. Therefore
we may reconstruct $\PrTop{\ThCat}$ as two different open subtopoi of
$\XTop$:
\[
  \begin{tikzpicture}[diagram]
    \node (2Th) {$2\cdot\PrTop{\ThCat}$};
    \node (Th/l) [above left = of 2Th,xshift=-.8cm] {$\PrTop{\ThCat}$};
    \node (Th/r) [below left = of 2Th,xshift=-.8cm] {$\PrTop{\ThCat}$};
    \node (X) [right = 2cm of 2Th] {$\XTop$};
    \path[open embedding*] (Th/l) edge node [desc,sloped] {$\mathsf{inj}_0$} (2Th);
    \path[open embedding] (Th/r) edge node [desc,sloped] {$\mathsf{inj}_1$} (2Th);
    \path[open embedding] (2Th) edge node [upright desc] {$\OpEmb$} (X);
    \path[open embedding,exists,bend left=25] (Th/l) edge node [above,sloped] {$\LOpEmb$} (X);
    \path[open embedding*,exists,bend right=25] (Th/r) edge node [below,sloped] {$\ROpEmb$} (X);
  \end{tikzpicture}
\]

We associate to each open subtopos of $\XTop$ a subterminal object and a
corresponding open modality in $\Sh{\XTop}$. In particular, we have opens
$\Mor|>->|{\SynOpn,\LOpn,\ROpn}{\ObjTerm{\Sh{\XTop}}}$ reconstructing
$\Psh{\ThCat}^2$ as $\Sl{\Sh{\XTop}}{\SynOpn}$, and $\Psh{\ThCat}$ twice as
$\Sl{\Sh{\XTop}}{\LOpn}$ and $\Sl{\Sh{\XTop}}{\ROpn}$ respectively,
corresponding to the symmetry of swapping the left and right syntactic
components of a parametricity structure. Moreover, $\SynOpn = \LOpn\lor\ROpn$
and $\LOpn\land\ROpn = \bot$.

Working synthetically, we may use the modalities $\MSynL,\MSynR$ in the
internal language of $\Sh{\XTop}$ to isolate the (left, right) syntactic parts
of a parametricity structure --- or to \emph{construct} parametricity
structures that are degenerate everywhere except for in their (left, right)
syntactic parts. The modality $\MSyn$ isolates the left and right parts of the
syntax together, and its \emph{closed complement} $\MSem$ is used to trivialize
the syntactic parts and isolate the semantic part: in particular, we have
$\MSyn{\MSem{X}} = \ObjTerm{}$.  The closed complement to an open modality is
not in general open, but it is always a \emph{lex idempotent modality} in the
sense of \citet{rijke-shulman-spitters:2017}.

The parametricity structure of phase separation is also expressed as an open
modality.  Recalling that we already have an open $\StOpn \in
\Opn{2\cdot\PrTop{\ThCat}}$ that isolates the static part of (each copy of) the
syntax, we note that we have an analogous open $\brc{\PtO}\in\Opn{\SP}$ of the
Sierpi\'nski topos that spans the open point $\PtO\in\SP$; by intersection, we
may therefore define an open of $\XTop$ to isolate the static part of a general
parametricity structure all at once: $ \StOpn := \OpEmb_*\StOpn \land
\ClEmb_*\brc{\PtO} $.

\begin{lemma}\label{lem:xtop-presheaf}
  The logos of parametricity structures $\Sh{\XTop}$ is a category of presheaves, \ie there exists a category $\DCat$ such that $\Sh{\XTop}\simeq \Psh{\DCat}$.
\end{lemma}

\begin{proof}
  First, we note that $\Psh{\ThCat}^2$ is $\Psh{2\cdot\ThCat}$ and $\Sh{\SP}$
  is $\Psh{\Simplex{1}}$. Moreover, the direct image
  $\Mor[\FR_*]{\Psh{\ThCat}^2}{\Sh{\SP}}$ is continuous, being a right adjoint; but this is
  one of the equivalent conditions for the stability of presheaf topoi under
  gluing identified by the Grothendieck school in SGA~4, Tome 1, Expos\'e iv,
  Exercise 9.5.10 (and worked out by \citet{carboni-johnstone:1995}).
\end{proof}

Consequently, we may construct $\Sh{\XTop}$ such that its internal dependent
type theory contains a \emph{strict} hierarchy of universes $\Univ{\alpha}$
\`a la \citet{hofmann-streicher:1997} and moreover enjoys the
\emph{strictification} axiom of \citet{orton-pitts:2016}, restated here as
\cref{axiom:strictification}. This is of course only possible because the
high-altitude structure of our work respects the principle of equivalence.

The central theorem of this section is an immediate consequence of the forgoing
discussion, combined with standard results in the presheaf semantics of dependent
type theory~\citep{hofmann-streicher:1997,hofmann:1997,streicher:2005}.

\begin{theorem}\label{thm:model}
  The category of sheaves $\Sh{\XTop}$ admits the structure of a model of \PSTT{}.
\end{theorem}

Combined with the internal constructions in \cref{sec:paramtt}, we may simply
unfold definitions until we reach a proof-relevant and phase separated version
of Reynolds' abstraction theorem~\citep{reynolds:1983} in the context of
\ModTT{}.

\begin{corollary}[Generalized abstraction theorem]\label{thm:reynolds}
  Fix two families of signatures $\sigma,\tau:\Val\prn{\SigTp}\to\Sig$, and a
  closed module functor $V:\Val\prn{\Prod{x :
  \SigTp}{\Prod{\_:\sigma\prn{x}}{\tau\prn{x}}}}$, together with a pair of
  closed module values $U_i:\Val\prn{\sigma\prn{T_i}}$ for a pair of closed types $T_0,T_1:\Val\prn{\SigTp}$.
  Now, fix a family of $\alpha$-small sets $\tilde{T}$ indexed in the closed values of type
  $T_0\times T_1$; the interpretations of $\sigma,\tau$ induce a pair of families of
  phase separated sets
  $\bbrk{\sigma}\prn{\tilde{T}},\bbrk{\tau}\prn{\tilde{T}}$
  indexed in the closed values of $\sigma\prn{T_0}\times\sigma\prn{T_1}$ and
  $\tau\prn{T_0}\times\tau\prn{T_1}$ respectively.
  The \textbf{generalized abstraction theorem} states that we have a function of phase separated sets from
  $\bbrk{\sigma}\prn{\tilde{T}}\brk{U_0,U_1}$ to
  $\bbrk{\tau}\prn{\tilde{T}}\brk{V\prn{T_0,U_0},V\prn{T_1,U_1}}$, tracked by a function between the static components.
\end{corollary}

A further consequence of our abstraction theorem is that the static behavior of
a module functor on closed modules does not depend on its dynamic behavior.

\section{Conclusions and future work}

What is the relationship between programming languages and their module
systems? Often seen as a useful feature by which to extend a programming
language, we contrarily view a language of modules as the ``basis theory'' that
any given programming language ought to extend. To put it bluntly, a
programming language is a universe $\mathcal{L}$ in the module type theory, and
specific aspects (such as evaluation order) are mediated by the decoding
function $\IsSig{t:\mathcal{L}}{\SigDyn{t}}$ of the universe.

\subsection{Relaxing the static--dynamic phase distinction}

In the present version of \ModTT{} we chose to force all ``object language''
types to be purely dynamic, in the sense that $\SigDyn{t}$ always has a trivial
static component. This design, inspired by the actual behavior of ML languages
with weak structure sharing (SML '97, OCaml, and \OneML), is by no means
forced: by allowing types to classify values with non-trivial static
components, we could reconstruct the ``half-spectrum'' dependent types
available in current versions of Haskell~\citep{eisenberg:2016:thesis}.

Allowing programs to have a non-trivial static component is also necessary to
support \emph{abstraction} in the presence of applicative functors like
$\Con{MkSet}$, as pointed out by \citet{rossberg-russo-dreyer:2014}. Under the
current strong static--dynamic phase distinction, abstraction for applicative
functors can still be achieved by ``tainting'' every value declaration with an
abstract type component, but there is reason to be skeptical this is \emph{in
fact} more desirable than simply achieving abstraction directly from general
type dependency. In light of both Idris~2 and
Lean~4~\citep{brady:2021,de-moura-ullrich:2021}, it would indeed be very hard
to argue today that full-spectrum type dependency presents any unsurmountable
problems for compilation of general-purpose programming languages.

Neither does it appear forced that module \emph{commands} should be statically
connected (except simply to reproduce the behavior of existing ML languages);
the original difficulty inherent in the question of when two impure modules are
identified by sharing seems to be already resolved by the modal separation of
effects \`a la \citet{moggi:1991}. Decoupling static connectivity from
computational effects significantly simplifies the theory of program modules.  Future
ML languages may expose \emph{closed} modalities like $\MDyn$ to enable more
flexible and fine-grained imposition of non-interference.

\subsection{Let a hundred phase distinctions bloom!}

Taking Reynolds's dictum\footnote{``Type structure is a syntactic discipline
for enforcing levels of abstraction''~\citep{reynolds:1983}.} seriously, we
believe that the phase distinction is the prototype for any number of
\emph{levels of abstraction}, each corresponding to a different open modality.
The lax-modal separation of effects renders full type dependency quite
unproblematic, hence some of the original motivations for the static--dynamic
phase distinction may be weaker than previously thought. In contrast, the
concept of \emph{phase distinction} generally is more important than ever.

\begin{example}[Logical relations]
  In this paper we considered the phase distinction between \textbf{syntactic}
  and \textbf{semantic}, which allows one to prove parametricity results as
  well as other important metatheorems such as canonicity and
  normalization~\citep{sterling-angiuli:2021,gratzer:normalization:2021}.
\end{example}

\begin{example}[Type refinements]
  Type refinements \`a la \citet{mellies-zeilberger:2015} can be interpreted by
  a phase distinction between \textbf{computational} and \textbf{logical}. Type
  refinements differ from the built-in verification capabilities of type theory
  in that logical/specification-level code is guaranteed to not interfere with
  computational-level code --- even when the specification-level information is
  proof-relevant. The view of type refinements as a phase distinction is a
  compelling alternative to realizability-style accounts of program
  extraction~\citep{constable:1986}. Here, extraction is implemented
  \emph{internally} by the weakening substitution
  $\Mor{\Gamma,\GenericOpn}{\Gamma}$.
\end{example}

\begin{example}[Separate compilation]
  Modules with free variables ranging over their dependencies (referred to as
  \emph{units} by \citet{flatt-felleisen:1998,smlsc:2006}) are an
  attractive model of separate compilation: each module can be compiled
  independently of its dependencies, which are then linked by means of
  simultaneous substitution or \emph{cut}. Separate compilation has the
  side effect, however, of limiting the ability of the compiler to
  generate more efficient code by inlining.
  Short of abandoning separate compilation entirely \`a la
  MLton~\citep{weeks:2006}, one may consider the suggestion of
  \citet[\S~1.5.3]{stone:2000} and \citet{leroy:2000} to use value-sharing
  (singletons) to expose definitions for inlining, but this has the destructive
  effect of breaking all abstraction boundaries imposed intentionally by the
  programmer.
  We suggest introducing a phase distinction between ``compile-time'' and
  ``development-time'', exposing inlineable definitions along compile-time
  extents $\Compr{\sigma}{\GenericOpn{\mathsf{cmpl}}\hookrightarrow V}$.
\end{example}

\begin{example}[Information flow]
  Information flow calculi \`a la \citet{abadi:1999} can be interpreted by a
  phase distinction between \textbf{low} and \textbf{high security}.  The open
  modality $\Op\Sub{\ell}$ projects the data that is visible to clients with
  security clearance $\ell$; the closed modality $\Cl\Sub{\ell}$ hides
  information from clients with clearance $\ell$. Non-interference follows immediately from
  the laws of the closed modality.
\end{example}

\subsection{Formalization of parametricity theorems}

Our approach is firmly rooted within the tradition of logical frameworks and
categorical algebra, which has enabled us to reduce the
highly technical (and very syntactic) logical relations arguments of prior work
on modules to some trivial type theoretic arguments that are amenable to
formalization \`a la \citet{orton-pitts:2016}. Actually formalizing the axioms
of \PSTT{} in a proof assistant like Agda, Coq, or Lean is within reach, thanks to the
work of \citet{gilbert-cockx-sozeau-tabareau:2019}.

\subsection{Non-trivial computational effects}

Another area for future work is to instantiate \ModTT{} with non-trivial
effects, such as recursive types or higher-order store. These features, often
accounted for using step-indexing, will likely require relativizing the
construction of \PSTT{} (\cref{sec:topos}) from $\SET$ to a logos in which
domain equations can be solved.

\section*{Acknowledgments}

Thanks to Mathieu Anel, Carlo Angiuli, Steve Awodey, Karl Crary, Derek Dreyer, Daniel
Gratzer, Guillaume Munch-Maccagnoni, Gordon Plotkin, and Michael Shulman for
their advice and comments, and to Tristan Nguyen at AFOSR for support.

This work was supported in part by AFOSR under grants MURI FA9550-15-1-0053 and
FA9550-19-1-0216.  Any opinions, findings and conclusions or recommendations
expressed in this material are those of the authors and do not necessarily
reflect the views of the AFOSR.

\nocite{dreyer-harper-chakravarty-keller:2007}
\nocite{stone-harper:2000}
\nocite{harper-lillibridge:1994}

\nocite{dreyer:2005,dreyer:2007} %
\nocite{mogelberg-simpson:2007}
\nocite{berger-mellies-weber:2012}

\nocite{cardelli-leroy:1990,leroy:1996,leroy:1994}
\nocite{aspinall:1995}

\bibliographystyle{ACM-Reference-Format}
\bibliography{references/refs-bibtex,temp-refs}

\end{document}